\documentclass[lettersize,journal,twoside]{IEEEtran}

\usepackage{amsmath,amsfonts}
\usepackage{algorithmic}
\usepackage{array}
\usepackage[caption=false,font=normalsize,labelfont=sf,textfont=sf]{subfig}
\usepackage{textcomp}
\usepackage{stfloats}
\usepackage{url}
\usepackage{verbatim}
\usepackage{graphicx}
\def\BibTeX{{\rm B\kern-.05em{\sc i\kern-.025em b}\kern-.08em
    T\kern-.1667em\lower.7ex\hbox{E}\kern-.125emX}}
\usepackage{balance}

\usepackage{multirow}
\usepackage{booktabs}

\usepackage{graphics}
\usepackage{epsfig} 
\usepackage{xcolor}
\usepackage{amssymb}  
\usepackage{amsthm}
\usepackage{euscript}
\usepackage{xcolor}
\usepackage{tasks}
\usepackage{hyperref}
\usepackage{paralist}
\usepackage{comment}
\usepackage{amsmath} 
{
\newtheorem{theorem}{\bf Theorem}
\newtheorem{lemma}{\bf Lemma}

\newtheorem{problem}{\bf Problem}

\newtheorem{assumption}{\bf Assumption}
\newtheorem{remark}{\bf Remark}

}

\newcommand{\trans}{\cdot}

\newcommand{\vect}[1]{\boldsymbol{#1}}
\newcommand{\bvect}[1]{\Bar{\boldsymbol{#1}}}

\renewcommand{\Bar}[1]{#1}

\newcommand{\GR}[1]{#1}

\begin{document}

\pagestyle{empty}

\title{Safe whole-body backstepping control for quadcopter path-following}

\author{Arthur H. D. Nunes, Arthur Da C. Vangasse, Guilherme V. Raffo,\\Vinicius M. Gon\c{c}alves, and Luciano C. A. Pimenta%
\thanks{The authors are with the Graduate Program in Electrial Engineering (PPGEE) of the Universidade Federal de Minas Gerais (UFMG), Belo Horizonte, MG, 31270-901, Brazil (emails: arthurhdn@ufmg.br; vangasse@ufmg.br; raffo@ufmg.br; leotorres@ufmg.br; lucpim@cpdee.ufmg.br). The first author is affiliated with NUAR Systems LTDA, Contagem, MG, Brazil.}%
\thanks{This work was in part supported by the grants CNPq (Brazilian National Research Council) 309925/2023-1 and 301298/2025-4, and FAPEMIG (under grants APQ-02144-18 and APQ-0063023). This work was partially supported by Petrobras/ANP under grants 2023/00494-5 and 2023/00643-0.}%
}



\maketitle
\thispagestyle{empty}

\begin{abstract}
This paper presents a novel whole-body Backstepping control strategy for safe quadcopter path-following. The proposed approach introduces an integrated control scheme that combines a translational guidance controller with a rigid-body attitude controller. To guarantee asymptotic path convergence, the method utilizes a nominal Integrated Guidance and Control (IGC) based on Artificial Vector Fields (AVF). To ensure reactive safety and collision avoidance, the control law is modified using a smooth distance function within the High-Order Control Barrier Function (HOCBF) framework. The quadcopter dynamics are modeled using quaternion algebra to represent position, velocity, and attitude. By combining the Backstepping approach with HOCBF, the controller guarantees that the vehicle avoids obstacle sets while successfully converging to the target path when unobstructed. The proposed methodology is validated through software-in-the-loop simulations and real-world experimental results using the Crazyflie platform.
\end{abstract}


\section{Introduction}

\IEEEPARstart{A}{utonomous} navigation in dynamic environments \GR{represents} a timely and important challenge in robotics and control. \GR{In particular}, Unmanned Aerial Vehicles (UAVs), such as quadcopters, have \GR{gained widespread adoption across} a variety of applications, including surveillance, \GR{environmental} monitoring, search and rescue, and autonomous delivery \cite{8384677,7487716,8714418,miranda2022autonomous}. 
%
%
For many of these tasks, path-following strategies are \GR{often} preferred over strict trajectory-tracking approaches, as they generally yield smoother convergence to the target path \GR{while requiring} less aggressive control efforts \cite{rubi2020survey}. 

To address the path-following problem, Artificial Vector Fields (AVFs) have been widely adopted as a guidance strategy \cite{yao2022guiding,yao2021singularity}. Constructive AVF methods \cite{rezende2021constructive} successfully \GR{guide} robots along dynamic, time-varying curves in $n$-dimensional spaces \cite{rezende2021constructive}. Furthermore, obstacle avoidance has been integrated into these methods by introducing reactive fields that circumnavigate the obstacle when a collision risk is detected \cite{nunes2022vector}. Despite their effectiveness, a major limitation of many traditional reactive avoidance approaches is their reliance on standard Euclidean distance function. Because the function may not be differentiable \GR{everywhere}, they are often non-smooth and result in discontinuous vector field control laws.
Such discontinuous guidance laws implicitly demand infinite acceleration, which cannot be achieved by actual rigid-body dynamics. 

To bridge this gap, recent methods \GR{have proposed} smooth distance functions \GR{alongside} Control Barrier Functions (CBFs) to guarantee reactive safety and collision avoidance \GR{within} a continuous control \GR{framework} \cite{ferreira2026pointtocloudnmpcsmoothavoidance,11358638,10530410,nunes2026safe}.

While high-level guidance algorithms dictate the desired kinematics, they must be coupled with low-level controllers to handle the underactuated rigid-body dynamics of the vehicle. Traditionally, this is accomplished using a cascaded architecture where the guidance layer is \GR{decoupled} from the control layer \cite{raffo2010integral,rezende2020robust,pereira2021collision}. \GR{Standard} linear controllers\GR{,} such as PID and LQR \cite{ozbek2016feedback,bouabdallah2004pid,cowling2007prototype}\GR{,} are frequently employed \GR{in these schemes} but struggle to maintain adequate performance during aggressive maneuvers. To overcome these limitations, nonlinear control techniques like Backstepping \cite{khalil2002nonlinear} have proven highly effective for quadcopters \cite{raffo2015robust,salierno2018whole,TAN201647,5428769,5531424}. However, \GR{existing literature predominantly employees} Backstepping \GR{within} a trajectory-tracking \GR{context}. 

In recent \GR{works} \cite{nunes2023integrated,nunes2026saferobusttubebasedpathfollowing}, Backstepping \GR{has been employed for path-following control} by integrating the AVF \cite{nunes2022vector}. \GR{While building upon a similar concept}, our \GR{approach incorporates the vehicle's} whole-body dynamics, \GR{enabling} formal proofs for \GR{both} safety and \GR{full-}system convergence. Furthermore, the \GR{proposed} method \GR{leverages a} smooth distance function \cite{nunes2026safe} \GR{within the HOCBF framework} \cite{8796030}.
To the best of the author's knowledge, this is the first work \GR{to achieve safe} path-following with obstacle avoidance \GR{by unifying} whole-body Backstepping control\GR{,} Artificial Vector Fields\GR{, and HOCBFs} for quadcopters\GR{, backed by convergence and safety} formal proofs \GR{under the complete dynamic model} for time-varying paths.


\GR{The main contributions of this work are summarized as follows:
\begin{itemize}    
    \item A novel whole-body path-following control scheme that seamlessly integrates translational guidance with quadcopter rigid-body dynamics using a Backstepping approach;
    \item An extension of a smooth distance function into a High-Order Control Barrier Function framework to guarantee reactive safety and collision avoidance;
    \item Formal mathematical proofs establishing closed-loop asymptotic path-following convergence and forward invariance (safety) for the full non-linear dynamic model;
    \item Experimental validation demonstrating the real-time performance and effectiveness of the proposed strategy through numerical simulations and real-world flight tests on a quadcopter platform.
\end{itemize}}

\subsection*{\GR{Notation}}
In this work\GR{,} we use quaternion algebra to describe the \GR{quadcopter} dynamics. Let $\mathbb{H}$ \GR{denote the set of} quaternion\GR{, defined as $
    \mathbb{H} = \{\boldsymbol{h} = h_0 + \hat{i}h_1 + \hat{j}h_2 + \hat{k}h_3: h_0,h_1,h_2,h_3 \in \mathbb{R} \}$,}
where $\hat{i}^2 = \hat{j}^2 = \hat{k}^2 = \hat{i}\hat{j}\hat{k} = -1$ are imaginary identities. 
The real and imaginary operators for quaternions are defined by $\mathrm{Re}(\boldsymbol{h}) = h_0$ and $\mathrm{Im}(\boldsymbol{h}) = \hat{i}h_1 + \hat{j}h_2 + \hat{k}h_3$, respectively\GR{, with the quaternion} conjugate \GR{given by} $\boldsymbol{h}^* = \mathrm{Re}(\boldsymbol{h}) - \mathrm{Im}(\boldsymbol{h})$. 
Let $\mathbb{H}_p = \{ \boldsymbol{h} \in \mathbb{H}: \mathrm{Re}(\boldsymbol{h}) = 0 \}$ \GR{represent} the set of \GR{pure} quaternions \GR{(i.e. quaternions with null real part).}  
This set \GR{forms a} 3-dimensional real vector subspace isomorphic to $\mathbb{R}^3$. We \GR{exploit this property} to represent \GR{positions} and velocities as pure quaternions, while reserving unit-norm quaternions \GR{to concisely represent} attitude. \GR{Within this framework}, the cross and dot products between two pure quaternions \GR{are defined} as $\boldsymbol{u} \times \boldsymbol{v} = \frac{\boldsymbol{u}\boldsymbol{v} - \boldsymbol{v}\boldsymbol{u}}{2}$ and $\boldsymbol{u} \cdot \boldsymbol{v} = -\frac{\boldsymbol{u}\boldsymbol{v} + \boldsymbol{v}\boldsymbol{u}}{2}$, respectively.
%
%
See \cite{silva2022dynamics,adorno2017robot,abaunza2017dual} for more details.

\section{System modeling}\label{sec:system_modeling}

The vehicle \GR{considered in} this work \GR{is a quadcopter featuring} four rotors symmetrically arranged in a \GR{single} plane. Each rotor generates \GR{an individual} force\GR{, and their} combination \GR{yields} the total thrust $u_T \in \mathbb{R}$ and the 3D \GR{control} torques $\boldsymbol{u}_\tau \in \mathbb{H}_p$, \GR{which constitute} the vehicle's control inputs \cite{raffo2010integral}. 


To \GR{formulate the system dynamics, let} $\boldsymbol{p} \in \mathbb{H}_p$ \GR{denote} the position and $\boldsymbol{v} \in \mathbb{H}_p$ the linear velocity \GR{in the inertial frame. Furthermore, let} $\boldsymbol{o} \in \mathbb{H}$ \GR{with}  $\|\boldsymbol{o}\| = 1$ \GR{represent the unit quaternion attitude,} and $\boldsymbol{\omega} \in \mathbb{H}_p$ \GR{be} the angular velocity \GR{in the body frame}. 
Given the states $\boldsymbol{p}, \boldsymbol{v}, \boldsymbol{o}$, and $\boldsymbol{\omega}$, \GR{the nominal} dynamic model \GR{of} the vehicle is given by
%
\begin{align}
    \Dot{\Bar{\boldsymbol{p}}} & =  \Bar{\boldsymbol{v}},\label{eq:nominal-model1}\\
    \Dot{\Bar{\boldsymbol{v}}} & =  -g\hat{\boldsymbol{k}} + m^{-1}\Bar{u}_T\Hat{\boldsymbol{k}}_{\Bar{B}}, \quad \Hat{\boldsymbol{k}}_{\Bar{B}} = \Bar{\boldsymbol{o}}\Hat{\boldsymbol{k}}\Bar{\boldsymbol{o}}^*,\label{eq:nominal-model2}\\
    \Dot{\Bar{\boldsymbol{o}}} & =  \frac{1}{2}\Bar{\boldsymbol{o}}\Bar{\boldsymbol{\omega}}, \label{eq:nominal-model3}\\
     \Dot{\Bar{\boldsymbol{\omega}}} & =  -J^{-1}(\Bar{\boldsymbol{\omega}} \times J\Bar{\boldsymbol{\omega}}) + J^{-1} \Bar{\boldsymbol{u}}_\tau.\label{eq:nominal-model4}
\end{align}
where $m \in \mathbb{R}$ \GR{denotes the total mass of the quadcopter}, $J \in \mathbb{R}^{3 \times 3}$ \GR{is its} inertia tensor\GR{, $g$ is the gravitational acceleration, and $\hat{\boldsymbol{k}}$ represents the vertical unit vector}.




\section{Problem statement}

The primary \GR{objective} considered in this work is to \GR{design a control law that drives} the vehicle to follow a \GR{predefined} target \GR{path,} as formalized next. 
\begin{problem}[Path-following, adapted from \cite{nunes2026saferobusttubebasedpathfollowing}]
    %
    Consider the system \eqref{eq:nominal-model1}-\eqref{eq:nominal-model4} and a time-varying target path $\vect{\mathcal{C}}(t)$ \GR{parametrized by} $\vect{c}(s,t)$. 
    \GR{The control objective is to design} a control law for $u_T$ and $\vect{u}_\tau$ such that the \GR{vehicle position ${\vect{p}}$ converges to and follows} the target path $\vect{\mathcal{C}}(t)$ \GR{in unobstructed environments}.
    \label{problem:path_following2}
\end{problem}

When the \GR{target} path is \GR{obstructed} and a collision is at risk, the controller must prioritize safety over \GR{path-}following. \GR{Consequently,} the \GR{vehicle} is allowed to \GR{deviate} from the curve to accomplish \GR{obstacle avoidance, resuming} the primary \GR{path-following objective} once the \GR{obstacle is cleared}. 
To \GR{formalize} the safety \GR{requirement}, let $\vect{\mathcal{O}}(t) \subset \mathbb{H}_p$ \GR{represent} a \GR{time-varying} 
obstacle set. \GR{We assume} that this set is represented by a finite number of time-varying points \(\vect{o}_i(t) \in \vect{\mathcal{O}}(t)\), in which the \GR{total} number of points \(N\) \GR{remains constant}, and each function \(\vect{o}_i(t)\) is differentiable in \(t\) \cite{nunes2026safe,nunes2026saferobusttubebasedpathfollowing}.

Let the half-squared distance between the point $\vect{p}$ and the set $\vect{\mathcal{O}}(t)$ be defined as
\begin{equation}
D_{\vect{\mathcal{O}}}(\vect{p},t) \equiv \min_{i=1,...,N}\frac{1}{2}\|\vect{p} - \vect{o}_i(t)\|^2.
\label{eq:half_squared_distance}
\end{equation}
%
Then, this problem is defined as follows.
\begin{problem}[Safety, adapted from \cite{nunes2026safe}]
    The vehicle must keep at least a safe distance $\lambda>0$ from the obstacle set, \emph{i.e.} 
    \begin{align}
        \vect{p}(t) & \in \vect{\mathcal{S}}(t) \quad \forall t \geq 0,
        \label{eq:condition_collision_avoidance}
    \end{align}
   where $\vect{\mathcal{S}}(t)$ is a safety set defined using barrier certificates \cite{8796030}
    \begin{align}
        \vect{\mathcal{S}}(t) = & \{ \vect{p} \in \mathbb{H}_p \ | \ B^\lambda(\vect{p},t) \geq 0\},\label{eq:safety_set}\\
        B^\lambda(\vect{p},t) & \equiv D_{\vect{\mathcal{O}}}(\vect{p},t) - \frac{\lambda^2}{2}.\label{eq:safety_set_barrier_function}
    \end{align}
    \label{problem:collision_avoidance}
\end{problem}

\GR{The core strategy is to prioritize Problem \ref{problem:collision_avoidance} at all times, addressing Problem \ref{problem:path_following2} whenever the vehicle is in a safe state.}


\section{Safe Backstepping Control}



\subsection{Path-following}\label{subsec:backstepping}

We first \GR{address} path-following in obstacle-free \GR{scenarios}. \GR{To this end, we} adopt a Backstepping \GR{strategy that directly integrates the} Artificial Vector Field guidance \GR{scheme} from \cite{rezende2021constructive} into the \GR{control formulation}.

\subsubsection{Step 1 \GR{- Kinematic Guidance Law}}
In \GR{the first} step, consider \GR{the translational kinematics} $\Dot{\Bar{\vect{p}}} = \Bar{\vect{v}}$, where $\Bar{\vect{v}}$ \GR{acts as a virtual control input}. This virtual control law \GR{is designed according to} the AVF \GR{methodology proposed in} \cite{rezende2021constructive}. 

The vector field virtual control law for $\Bar{\vect{v}}$\GR{, denoted as} $\vect{\Phi}(\Bar{\vect{p}},t)$\GR{, is computed by identifying} the closest point \GR{on} the path, \GR{its distance to the vehicle}, and tangent vectors\GR{, as follows}%
\begin{align}
    s_*(\bvect{p},t) & = \arg \min_s \| \bvect{p} - \vect{c}(s,t) \|,\label{eq:s_star} \\
     \vect{c}_*(\bvect{p},t) & = \vect{c}(s_*(\bvect{p},t),t),\label{eq:c_star}\\
     \vect{D}_{\vect{\mathcal{C}}}(\bvect{p},t) & = \bvect{p} - \vect{c}_*(\bvect{p},t), \label{eq:curve_distance_vector} \\
     \vect{T}_{\vect{\mathcal{C}}}(\bvect{p},t) & = \dfrac{\partial \vect{c}(s,t)}{\partial s}\biggr\rvert_{s = s_*}.\label{eq:field_T}
\end{align}


Two vectors \GR{representing} the convergence and circulation behaviors of the field are weighted by scalar gains $G(\|\vect{D}_{\vect{\mathcal{C}}}\|) = \frac{2}{\pi}\tan^{-1}(k_G\|\vect{D}_{\vect{\mathcal{C}}}\|), \ k_G > 0,$ and $H(\|\vect{D}_{\vect{\mathcal{C}}}\|) = \pm\sqrt{1 - \left[G(\|\vect{D}_{\vect{\mathcal{C}}}\|)\right]^2}$, respectively.
Then, we compute a static unit field and a feed-forward component, \GR{which are given by}
\begin{align}
        \vect{\Phi}_S & = -G\frac{\vect{D}_{\vect{\mathcal{C}}}}{\|\vect{D}_{\vect{\mathcal{C}}}\|} + H\frac{\vect{T}_{\vect{\mathcal{C}}}}{\|\vect{T}_{\vect{\mathcal{C}}}\|}, \quad
        \boldsymbol{\Phi}_T = -\Pi_T\frac{\partial \boldsymbol{D}_{\vect{\mathcal{C}}}}{\partial t},
    \label{eq:static_and_time_field}
\end{align}
\GR{with $\Pi_T$ being} the null space projector of $\vect{T}_{\vect{\mathcal{C}}}$.

To weight the sum of the components defined above, \GR{an additional} scalar gain $\eta(\bvect{p},t) \in [0, v_r]$ is \GR{introduced} to ensure \GR{that} $\|\vect{\Phi}(\bvect{p},t)\| = v_r$\GR{,} which is a \GR{specified} reference speed. This gain is given by
\begin{gather}
    \eta = -\vect{\Phi}_S \cdot \vect{\Phi}_T + \sqrt{(\vect{\Phi}_S \cdot \vect{\Phi}_T)^2 + v_r^2 - \|\vect{\Phi}_T\|^2}.
\end{gather}
\GR{As established in} \cite{rezende2021constructive}, $\eta$ is \GR{well-}defined for \GR{paths whose speed is} slower than the \GR{vehicle's} velocity, \GR{i.e.} $\|\boldsymbol{\Phi}_T\|<v_r$.

Finally, the field $\vect{\Phi}(\bvect{p},t)$ is given by
\begin{equation}
    \begin{array}{rcl}
        \vect{\Phi}(\bvect{p},t) & = & \eta(\bvect{p},t)\vect{\Phi}_S(\bvect{p},t) + \vect{\Phi}_T(\bvect{p},t).
    \end{array}
    \label{eq:field}
\end{equation}

\begin{lemma}[Adapted from Proposition 2 of \cite{rezende2021constructive}]
    For the nominal system \eqref{eq:nominal-model1}, if $\bvect{v} = \vect{\Phi}(\bvect{p},t)$, the trajectories of $\bvect{p}$ converge to and follow the target curve $\vect{\mathcal{C}}(t)$.
    \label{lemma:vector_field_convergence}
\end{lemma}
\begin{proof}
    Let $V_p: \mathbb{H}_p \rightarrow \mathbb{R}$ be a candidate Lyapunov function, given by 
    \begin{equation}
        V_p(\boldsymbol{D}_{\boldsymbol{\mathcal{C}}}) = \frac{1}{2} \boldsymbol{D}_{\boldsymbol{\mathcal{C}}}\trans \boldsymbol{D}_{\boldsymbol{\mathcal{C}}},\label{eq:V_p_defintion}
    \end{equation}
    with $\boldsymbol{D}_{\boldsymbol{\mathcal{C}}}$ \GR{defined in} \eqref{eq:curve_distance_vector}.
    According to Proposition 2 of \cite{rezende2021constructive}, when $\bvect{v} = \vect{\Phi}(\bvect{p},t)$, the following result holds:
    \begin{align}
    \Dot{V}_{p}(\boldsymbol{D}_{\boldsymbol{\mathcal{C}}}) & = -\eta G \|\boldsymbol{D}_{\boldsymbol{\mathcal{C}}}\|.\label{eq:Dot_V_p_result}
    \end{align}
Since $\Dot{V}_{p}(\boldsymbol{D}_{\boldsymbol{\mathcal{C}}})$ is negative definite, then $\boldsymbol{D}_{\boldsymbol{\mathcal{C}}} \rightarrow 0$ as $t\rightarrow\infty$, \GR{ensuring that the system} trajectories converge \GR{toward} the path. As $\boldsymbol{D}_{\boldsymbol{\mathcal{C}}} \rightarrow 0$, the circulation component in \eqref{eq:static_and_time_field} becomes dominant\GR{, causing the vehicle to follow the path}.
\end{proof}

\subsubsection{Step 2 \GR{ - Acceleration Control Law}}
\GR{In} the second step, \GR{define the velocity tracking error as} $\vect{z}_v \equiv \Bar{\vect{v}} - \vect{\Phi}$\GR{,} and consider \GR{a virtual acceleration control input} $\Bar{\boldsymbol{a}}$ such that $\Dot{\Bar{\boldsymbol{v}}} = -g \Hat{\boldsymbol{k}} + \Bar{\boldsymbol{a}}$. The \GR{resulting} augmented \GR{translational error dynamics are}
\begin{align}
    \Dot{\vect{D}}_{\vect{\mathcal{C}}} & = \vect{\Phi} - \Dot{\vect{c}}_* + \vect{z}_v,\\
    \Dot{\vect{z}}_v & = -g \Hat{\boldsymbol{k}} + \Bar{\vect{a}} - \Dot{\vect{\Phi}}.
\end{align}

\GR{To stabilize the velocity error and track the guidance vector field, we propose the following virtual acceleration control law:}
\begin{align}
    \Bar{\vect{a}}_{\vect{\Phi}} & = g \Hat{\boldsymbol{k}} - k_p\vect{D}_{\vect{\mathcal{C}}} - k_v\vect{z}_v + \Dot{\vect{\Phi}}, \quad k_p, k_v > 0.
    \label{eq:control_law_acceleration_curve_following}
\end{align}
\begin{lemma}
    \GR{Consider the translational dynamics} \eqref{eq:nominal-model1}-\eqref{eq:nominal-model2}\GR{. If} the virtual control law $\frac{\Bar{u}_T}{m}\Hat{\boldsymbol{k}}_{\Bar{B}} \equiv \Bar{\vect{a}} = \Bar{\vect{a}}_{\vect{\Phi}}$ from \eqref{eq:control_law_acceleration_curve_following} \GR{is applied, then} the \GR{system} trajectories \GR{asymptotically} converge to and follow the target curve $\vect{\mathcal{C}}(t)$.
    \label{lemma:nominal_path_following}
\end{lemma}
\begin{proof}
    Let $V_{p,v}: \mathbb{H}_p \times \mathbb{H}_p \rightarrow \mathbb{R}$ be a candidate Lyapunov function \GR{defined as}
    \begin{align}
    V_{p,v}(\vect{D}_{\vect{\mathcal{C}}},\vect{z}_v) & = k_pV_p + \frac{1}{2}\vect{z}_v\trans\vect{z}_v, \quad k_p > 0,\label{eq:V_p_v_defintion}
\end{align}
with $V_p$ \GR{defined in} \eqref{eq:V_p_defintion}.

\GR{Taking the time derivative of $V_{p,v}$ yields}
\begin{align}
    \Dot{V}_{p,v} & = -k_p\eta G \|\vect{D}_{\vect{\mathcal{C}}}\| + k_p\vect{D}_{\vect{\mathcal{C}}}\trans\vect{z}_v + \vect{z}_v\trans\left( -g \Hat{\boldsymbol{k}} + \Bar{\vect{a}} - \Dot{\vect{\Phi}}\right).
\end{align}

\GR{Substituting the nominal control law} $\Bar{\vect{a}} = \Bar{\vect{a}}_{\vect{\Phi}}$ \GR{from} \eqref{eq:control_law_acceleration_curve_following} \GR{gives}
\begin{align}
    \Dot{V}_{p,v} & = -k_p\eta G \|\vect{D}_{\vect{\mathcal{C}}}\| - k_v\|\vect{z}_v\|^2.
\end{align}
\GR{Since $\dot{V}_{p,v}$ is negative definite with respect to the states $(\boldsymbol{D}_{\mathcal{C}}, \boldsymbol{z}_v)$, this implies that $\boldsymbol{D}_{\mathcal{C}} \to 0$ and $\boldsymbol{z}_v \to 0$ as $t \to \infty$. Consequently, system trajectories asymptotically converge to and follow the target path according to Lemma~\ref{lemma:vector_field_convergence}.}
\end{proof}

This step \GR{yielded} a virtual \GR{acceleration} control law $\Bar{\vect{a}}$. \GR{Within the translational dynamic model \eqref{eq:nominal-model2}, the vehicle's actual acceleration component is given by} $m^{-1}\Bar{u}_T\Hat{\boldsymbol{k}}_{\Bar{B}}$. \GR{At this stage, one could define the thrust control law $T_0$ as}
\begin{align}
    T_0 & = m ( \Bar{\vect{a}}_{\vect{\Phi}} \trans \Hat{\boldsymbol{k}}_{\Bar{B}}).
\end{align}
\GR{While this choice is suitable when addressing only the nominal path-following Problem~\ref{problem:path_following2}, directly setting $\Bar{u}_T = T_0$ is insufficient when incorporating safety guarantees. To simultaneously ensure collision avoidance, our strategy does not control ${u}_T$ directly. Instead, we initialize $\Bar{u}_T(t_0) = T_0(t_0)$ and design a control law for its second time derivative $\Ddot{\Bar{u}}_T$.} 

\subsubsection{Step 3 \GR{- Attitude Kinematic Control Law}}
\GR{In the third step, define the acceleration error as} $\vect{z}_o \equiv \frac{\Bar{u}_T}{m}\Hat{\boldsymbol{k}}_{\Bar{B}} - \Bar{\vect{a}}_{\vect{\Phi}}$. 
%
\GR{Next, consider the angular velocity} $\Bar{\boldsymbol{\omega}}$ and \GR{the thrust-rate} $\Dot{\Bar{u}}_T$ as virtual control \GR{inputs}. \GR{Then,} the augmented system \GR{can be written as}
\begin{align}
    \Dot{\vect{D}}_{\vect{\mathcal{C}}} & = \vect{\Phi} - \Dot{\vect{c}}_* + \vect{z}_v,\\
    \Dot{\vect{z}}_v & =  - k_p\vect{D}_{\vect{\mathcal{C}}} - k_v\vect{z}_v + \vect{z}_o,\\
    \Dot{\vect{z}}_o & =  \Bar{\boldsymbol{o}}\left(\Bar{\boldsymbol{\omega}}\times \Hat{\boldsymbol{k}}\right)\Bar{\boldsymbol{o}}^*\frac{\Bar{u}_T}{m} + \Hat{\boldsymbol{k}}_{\Bar{B}}\frac{\Dot{\Bar{u}}_T}{m} - \Dot{\Bar{\vect{a}}}_{\vect{\Phi}}.
\end{align} 

To stabilize the acceleration error $\boldsymbol{z}_o$, we propose the following virtual control law $\Bar{\boldsymbol{\varphi}}$ for the angular velocity $\Bar{\boldsymbol{\omega}}$:
\begin{align}
    \Bar{\vect{\varphi}} & = \Hat{\boldsymbol{k}} \times \Bar{\boldsymbol{o}}^*\left({\Bar{\vect{\varphi}}_L}\frac{m}{\Bar{u}_T}\right)\Bar{\boldsymbol{o}} + \Hat{\boldsymbol{k}} \omega_{yaw},\label{eq:varphi_control_law}\\
    \Bar{\vect{\varphi}}_L & = \Dot{\Bar{\vect{a}}}_{\vect{\Phi}} - k_o \vect{z}_o - k_{p,v}\vect{z}_v,\quad k_o,k_{p,v}>0,\label{eq:varphi_L}
\end{align}
where $\omega_{yaw} \in \mathbb{R}$ is \GR{an auxiliary} control law \GR{available to independently specify the yaw rate}, as \GR{discussed} later. 

\begin{assumption}
    To ensure \GR{that} the virtual \GR{angular velocity} control law remains well-defined and singularity-free, the total thrust generated by the vehicle \GR{must remain} strictly positive throughout operation\GR{, i.e.,} $u_{T} > 0$ for all $t \ge 0$.
    \label{assumption:positive_uT}
\end{assumption}
\begin{remark}
\GR{This assumption is physically reasonable,} as quadcopters must continuously generate upward thrust to counteract the gravitational acceleration and maintain altitude. Under nominal path-following and collision avoidance scenarios, the vehicle \GR{operates} near hover or \GR{performs} forward flight, \GR{where} the thrust $u_{T}$ inherently \GR{fluctuates} around the vehicle's weight ($mg$). A scenario \GR{in which} $u_{T}=0$ corresponds to \GR{unpowered free-fall, under which} the quadcopter \GR{completely} loses control authority over its translational dynamics.
\end{remark}

The virtual \GR{thrust-rate} control law is \GR{defined as}
\begin{align}
    T_1 & = m(\Bar{\vect{\varphi}}_L\trans\Hat{\boldsymbol{k}}_{\Bar{B}}).\label{eq:thrut_rate_T1}
\end{align}


\begin{lemma}
     \GR{Consider} the \GR{dynamic} system  \eqref{eq:nominal-model1}-\eqref{eq:nominal-model3}\GR{. If} the virtual control laws $\Bar{\vect{\omega}} = \Bar{\vect{\varphi}}$ from \eqref{eq:varphi_control_law} and $\Dot{\Bar{u}}_T = T_1$ from \eqref{eq:thrut_rate_T1} \GR{are applied, then the system trajectories asymptotically} converge to and follow the target curve $\vect{\mathcal{C}}(t)$.
\end{lemma}
\begin{proof}
Let $V_{p,v,o}: \mathbb{H}_p \times \mathbb{H}_p \times \mathbb{H}_p \rightarrow \mathbb{R}$ be a candidate Lyapunov function \GR{defined as}
    \begin{align}
    V_{p,v,o}(\vect{D}_{\vect{\mathcal{C}}},\vect{z}_v, \vect{z}_o) & = k_{p,v}V_{p,v} + \frac{1}{2}\vect{z}_o\trans\vect{z}_o, \quad k_{p,v} > 0, \label{eq:V_p_v_o_defintion}
\end{align}
\GR{where} $V_{p,v}$ \GR{is given in} \eqref{eq:V_p_v_defintion}.

Its time derivative \GR{along the system trajectories} is
\begin{align}
    \Dot{V}_{p,v,o} & = k_{p,v}\left(-k_p\eta G \|\vect{D}_{\vect{\mathcal{C}}}\| - k_v\|\vect{z}_v\|^2 + \vect{z}_v\trans\vect{z}_o\right) + \vect{z}_o\trans\Dot{\vect{z}}_o.\label{eq:V_p_v_o_derivative}
\end{align}
\GR{Expanding the} term $\vect{z}_o\trans\Dot{\vect{z}}_o$ \GR{yields}
%
\begin{align}
    \vect{z}_o\trans\Dot{\vect{z}}_o & = \vect{z}_o\trans\left[(\Bar{\boldsymbol{o}}\left(\Bar{\boldsymbol{\omega}}\times \Hat{\boldsymbol{k}}\right)\Bar{\boldsymbol{o}}^*)\frac{\Bar{u}_T}{m} +\Hat{\boldsymbol{k}}_{\Bar{B}}\frac{\Dot{\Bar{u}}_T}{m}- \Dot{\Bar{\vect{a}}}_{\vect{\Phi}}\right].
\end{align}

\GR{Applying the angular velocity control law $\Bar{\vect{\varphi}}$ from \eqref{eq:varphi_control_law} for $\Bar{\vect{\omega}}$, we obtain}
\begin{align}
    &\vect{z}_o\trans\left[\left(\Hat{\boldsymbol{k}}_{\Bar{B}} \times \frac{\Bar{\vect{\varphi}}_L}{\frac{\Bar{u}_T}{m}} + \Hat{\boldsymbol{k}}_{\Bar{B}} \omega_{yaw}\right)\times\Hat{\boldsymbol{k}}_{\Bar{B}}\right]\frac{\Bar{u}_T}{m}\nonumber\\
    &+ \boldsymbol{z}_o\trans\Hat{\boldsymbol{k}}_{\Bar{B}}\frac{\Dot{\Bar{u}}_T}{m}- \vect{z}_o\trans\Dot{\Bar{\vect{a}}}_{\vect{\Phi}} .\label{eq:expanding_zo_dotzo_intermediate}
\end{align}
\GR{Since} $\Hat{\boldsymbol{k}}_{\Bar{B}} \omega_{yaw}\times\Hat{\boldsymbol{k}}_{\Bar{B}}$ is null, \GR{the yaw control input} $\omega_{yaw}$ \GR{vanishes, demonstrating that $\omega_{\mathrm{yaw}}$ can be selected independently. Using vector identities,} \eqref{eq:expanding_zo_dotzo_intermediate} \GR{simplifies to}
\begin{align}
        \vect{z}_o\trans\Dot{\vect{z}}_o & = -k_o\|\vect{z}_o\|^2 - k_{p,v}\vect{z}_v\trans\vect{z}_o \nonumber\\
        &+ (\vect{z}_o\trans\Hat{\boldsymbol{k}}_{\Bar{B}})\left(\frac{\Dot{\Bar{u}}_T}{m} - \Bar{\vect{\varphi}}_L\trans\Hat{\boldsymbol{k}}_{\Bar{B}}\right).
\end{align}

Notice that if $\Bar{u}_T$ was controlled directly \GR{as} $\Bar{u}_T = T_0$\GR{, then} $\vect{z}_o \trans \Hat{\boldsymbol{k}}_{\Bar{B}} = 0$ and the \GR{final} term would \GR{vanish}. \GR{In the} dynamically extended \GR{case}, \GR{setting} $\Dot{u}_T=T_1$ \GR{likewise} cancels \GR{this remaining} term. 
\GR{Substituting this result back into \eqref{eq:V_p_v_o_derivative} yields}
\begin{align}
    \Dot{V}_{p,v,o} & = -k_{p,v}k_p\eta G \|\vect{D}_{\vect{\mathcal{C}}}\| - k_{p,v}k_v\|\vect{z}_v\|^2 - k_o\|\vect{z}_o\|^2.
\end{align}
\GR{This confirms that $\dot{V}_{p,v,o}$ is negative definite, implying that $\boldsymbol{D}_{\boldsymbol{\mathcal{C}}} \to 0$, $\vect{z}_v \to 0$, and $\vect{z}_o \to 0$ as $t \to \infty$. Consequently, the system trajectories asymptotically converge to and track the desired path.}
\end{proof}


\subsubsection{Step 4 \GR{- Torque Control Law}}
In the fourth and \GR{final} step, \GR{define the angular velocity tracking error as} $\vect{z}_\omega \equiv \Bar{\vect{\omega}} - \Bar{\vect{\varphi}}$ and \GR{the thrust-rate error as} $z_T = \Dot{u}_T - T_1$. \GR{The complete augmented error dynamics are given by}
\begin{align}
    \Dot{\vect{D}}_{\vect{\mathcal{C}}} & = \vect{\Phi} - \Dot{\vect{c}}_* + \vect{z}_v,\\
    \Dot{\vect{z}}_v & =  - k_p\vect{D}_{\vect{\mathcal{C}}} - k_v\vect{z}_v + \vect{z}_o,\\
    \Dot{\vect{z}}_o & =  \boldsymbol{o}((\vect{\varphi}+\vect{z}_\omega)\times\Hat{\boldsymbol{k}})\boldsymbol{o}^*\frac{\Bar{u}_T}{m}+ \Hat{\boldsymbol{k}}_{\Bar{B}}\frac{T_1 + z_T}{m} - \Dot{\Bar{\vect{a}}}_{\vect{\Phi}},\\
    \Dot{\vect{z}}_\omega & = -J^{-1}(\Bar{\boldsymbol{\omega}} \times J\Bar{\boldsymbol{\omega}}) + J^{-1} \Bar{\boldsymbol{u}}_\tau - \Dot{\Bar{\vect{\varphi}}},\label{eq:dot_z_w}\\
    \Dot{z}_T & = \Ddot{u}_T - \Dot{T}_1.\label{eq:dot_z_T}
\end{align} 


The \GR{torque} control law $\Bar{\vect{u}}_\tau$ is \GR{designed as}
\begin{align}
    \vect{\tau}_\Phi & = \Bar{\boldsymbol{\omega}} \times J\Bar{\boldsymbol{\omega}} + J\Dot{\Bar{\vect{\varphi}}} - k_\omega\vect{z}_\omega - k_{p,v,o}\frac{\Bar{u}_T}{m}(\Hat{\boldsymbol{k}}\times\boldsymbol{o}^*\vect{z}_o\boldsymbol{o}),\label{eq:u_tau_bar}
\end{align}
\GR{where $k_\omega > 0$ is a positive gain matrix (or scalar).}

\GR{Similarly, the virtual control law for the second derivative of thrust, $\ddot{\Bar{u}}_T$, is given by}
\begin{align}
    T_{2,\Phi} & = \Dot{T}_1 - k_Tz_T - (\vect{z}_o\trans\Hat{\boldsymbol{k}}_{\Bar{B}})\frac{k_{p,v,o}}{m}, \quad k_T > 0.\label{eq:ddot_u_T_law}
\end{align}

\begin{theorem}[Path following]
    \GR{Consider the nominal quadrotor system dynamics} \eqref{eq:nominal-model1}-\eqref{eq:nominal-model4}\GR{. If Assumption \ref{assumption:positive_uT} holds and the feedback control laws} $\Bar{\vect{u}}_\tau = \vect{\tau}_\Phi$ from \eqref{eq:u_tau_bar} and $\Ddot{u}_T = T_{2,\Phi}$ from \eqref{eq:ddot_u_T_law} \GR{are applied, then the system} trajectories \GR{asymptotically} converge to and follow the target curve $\vect{\mathcal{C}}(t)$.
    \label{lemma:path_following}
\end{theorem}
\begin{proof}
\GR{Consider the candidate Lyapunov function} $V: \mathbb{H}_p \times \mathbb{H}_p \times \mathbb{H}_p \times \mathbb{H}_p \times \mathbb{R} \rightarrow \mathbb{R}$ \GR{defined as}
    \begin{align}
    V(\vect{D}_{\vect{\mathcal{C}}},\vect{z}_v, \vect{z}_o, \vect{z}_\omega, z_T) & = k_{p,v,o}V_{p,v,o}
     + \frac{1}{2}\vect{z}_\omega\trans J\vect{z}_\omega + \frac{1}{2}z_T^2, \label{eq:V_p_v_o_omega_defintion}
\end{align}
with $k_{p,v,o} > 0$ and $V_{p,v,o}$ \GR{being defined in} \eqref{eq:V_p_v_o_defintion}.

\GR{Taking the time derivative of $V$ along the system trajectories yields}
\begin{align}
    &\Dot{V} = -k_{p,v,o}k_{p,v}k_p\eta G \|\vect{D}_{\vect{\mathcal{C}}}\| - k_{p,v,o}k_{p,v}k_v\|\vect{z}_v\|^2 \nonumber\\
    &- k_{p,v,o}k_o\|\vect{z}_o\|^2 + k_{p,v,o}\frac{\Bar{u}_T}{m}\vect{z}_o\trans\boldsymbol{o}\left(\vect{z}_\omega\times\Hat{\boldsymbol{k}}\right)\boldsymbol{o}^* + \vect{z}_\omega\trans J\Dot{\vect{z}}_\omega\nonumber\\
    & + k_{p,v,o}(\vect{z}_o\trans\Hat{\boldsymbol{k}}_{\Bar{B}})\frac{z_T}{m} +z_T\Dot{z}_T .\label{eq:V_p_v_o_w_derivative}
\end{align}

Substituting \eqref{eq:u_tau_bar} into \eqref{eq:dot_z_w} \GR{gives}
\begin{align}
    \Dot{\vect{z}}_\omega & = - k_\omega J^{-1}\vect{z}_\omega - k_{p,v,o}J^{-1}\frac{\Bar{u}_T}{m}(\Hat{\boldsymbol{k}}\times\boldsymbol{o}^*\vect{z}_o\boldsymbol{o}).
\end{align}
\GR{Similarly,} substituting \eqref{eq:ddot_u_T_law} into \eqref{eq:dot_z_T} \GR{results in}
\begin{align}
    \Dot{z}_T & =  - k_Tz_T - (\vect{z}_o\trans\Hat{\boldsymbol{k}}_{\Bar{B}})\frac{k_{p,v,o}}{m}. 
\end{align}

\GR{Substituting these error dynamics back into \eqref{eq:V_p_v_o_w_derivative} leads to cross-term cancellations, simplifying the time derivative to}
\begin{align}
    \Dot{V} =& -k_{p,v,o}k_{p,v}k_p\eta G \|\vect{D}_{\vect{\mathcal{C}}}\| - k_{p,v,o}k_{p,v}k_v\|\vect{z}_v\|^2 \nonumber\\
    &- k_{p,v,o}k_o\|\vect{z}_o\|^2 - k_\omega\|\vect{z}_\omega\|^2 - k_Tz_T^2.
\end{align}

\GR{Since $\dot{V}$ is negative definite with respect to the states $(\boldsymbol{D}_{\mathcal{C}}, \boldsymbol{z}_v, \boldsymbol{z}_o, \boldsymbol{z}_\omega, z_T)$, this guarantees that $\boldsymbol{D}_{\mathcal{C}} \to \boldsymbol{0}$, $\boldsymbol{z}_v \to \boldsymbol{0}$, $\boldsymbol{z}_o \to \boldsymbol{0}$, $\boldsymbol{z}_\omega \to \boldsymbol{0}$, and $z_T \to 0$ as $t \to \infty$. Consequently, the system trajectories asymptotically converge to and follow the target curve $\vect{\mathcal{C}}(t)$.}
\end{proof}

\subsection{Collision avoidance}
\GR{In the next phase of our methodology,} we incorporate obstacle avoidance \GR{capabilities into the nominal quadrotor system} \eqref{eq:nominal-model1}-\eqref{eq:nominal-model4}. 
%
%
\GR{To synthesize} the safety control law, we use the smooth half-squared distance function $D_{\boldsymbol{\mathcal{O}}}^{h}(\boldsymbol{p},t)$, \GR{parameterized by} a smoothing parameter $h > 0$, \GR{as proposed by} \cite{nunes2026safe},
\begin{align}
    D_{\boldsymbol{\mathcal{O}}}^{h}(\boldsymbol{p},t)& \equiv \left(\sum_{i=1}^N
     \left(\frac{\|\boldsymbol{p} - \boldsymbol{o}_i(t)\|^2}{2}\right)^{-\frac{1}{h}}
     \right)^{-h}.
    \label{eq:smooth_distance_definition}
\end{align}
\GR{As demonstrated in} \cite{nunes2026safe}, this smooth function lower-bounds the Euclidean half-squared distance: $D_{\boldsymbol{\mathcal{O}}}^{h}(\boldsymbol{p},t) \le D_{\boldsymbol{\mathcal{O}}}(\boldsymbol{p},t)$.
\begin{assumption}
    As discussed in \cite{nunes2026safe}, \GR{let} $\boldsymbol{p}(t) \in \boldsymbol{\mathcal{D}}(t)$, where $\boldsymbol{\mathcal{D}}(t) = \{ \boldsymbol{p} \in \mathbb{R}^3|\frac{ \partial D_{\boldsymbol{\mathcal{O}}}^{h}(\boldsymbol{p},t)}{\partial \boldsymbol{p}} \neq \boldsymbol{0}\}$ \GR{defines the region} in the workspace \GR{where the gradient of} the smoothed distance vector is \GR{nonzero at time $t$}.
\label{assumption:smmothed_distance_vector_never_null}
\end{assumption}
\begin{remark}
The set of points where the gradient vanishes represents isolated equilibrium states and has zero measure. Consequently, this condition does not restrict the practical operation of the controller \cite{nunes2026safe}.
\end{remark}

Let \GR{$B^{h,\Bar{\lambda}}(\bvect{p},t) = D_{\vect{\mathcal{O}}}^h(\bvect{p},t) - \frac{\Bar{\lambda}^2}{2}$,} 
\GR{which extends} $B^{\lambda}(\vect{p},t)$ in \eqref{eq:safety_set_barrier_function} \GR{by employing the smooth distance function} \eqref{eq:smooth_distance_definition}. \GR{The objective of our collision avoidance strategy is to design control inputs that render the safe set forward invariant, ensuring} $B^{h,\Bar{\lambda}}(\bvect{p},t) \geq 0$, $\forall t \geq 0$.

\GR{To guarantee forward invariance of the safe set and ensure that the barrier condition is satisfied, the control \GR{laws} $\vect{u}_\tau$ and $\Ddot{u}_T$ must be designed accordingly. Because these control inputs first appear in the fourth time derivative of the barrier function, we resort to the HOCBF framework.}

Let \GR{$\psi_0 = B^{h,\Bar{\lambda}}(\bvect{p},t)$ and $\psi_i = \Dot{\psi}_{i-1} + \alpha_i(\psi_{i-1})$,} 
where $\alpha_i$ is a $\mathcal{K}$-class function.
\GR{To ensure safety while preserving nominal performance, we seek control inputs $\vect{u}_\tau$ and $\Ddot{u}_T$ that remain as close as possible to the nominal path-following control laws $\vect{\tau}_\Phi$ and $T_{2,\Phi}$, while satisfying the fourth-order barrier constraint $\psi_4(\vect{p},\vect{v},\vect{o},\vect{\omega}, \vect{u}_\tau, u_T, \Dot{u}_T, \Ddot{u}_T) \ge 0$ \cite{8796030}. This leads to the following point-wise quadratic optimization problem:}
\begin{equation}
    \vect{u}_{\tau,*},\Ddot{u}_{T,*} \in \arg \min_{\vect{\tau},T} \| \vect{\tau}_\Phi - \vect{\tau}\|^2 + \| T_{2,\Phi} - T \|^2, \  \mathrm{s.t.} \ \psi_4 \ge 0.
\end{equation}

\GR{To determine whether the nominal control inputs are safe, let $F$ denote the fourth-order barrier function $\psi_4$ evaluated at the nominal path-following control laws} $\vect{\tau}_\Phi$ and $T_{2,\Phi}$
\begin{align}
F = \psi_4(\vect{p},\vect{v},\vect{o},\vect{\omega}, \vect{\tau}_\Phi, u_T, \Dot{u}_T, T_{2,\Phi}).\label{eq:definition_F}
\end{align}
If $F \ge 0$, \GR{the nominal control inputs already satisfy the safety constraint and are directly applied. Conversely, if} $F<0$, the \GR{nominal} control laws are modified \GR{by} additive \GR{safe corrections $\boldsymbol{\tau}_{\Psi}$ and $T_{2,\Psi}$ designed to enforce the HOCBF is greater than zero}. 
This can be guaranteed \GR{by}
\begin{align}
    \vect{\tau}_{\Psi} & = \frac{-F}{A_T^2 + \|\vect{A}_\tau\|^2}\vect{A}_\tau, \quad T_{2,\Psi} = \frac{-F}{A_T^2 + \|\vect{A}_\tau\|^2}A_T,\label{eq:u_psi}
\end{align}
where {$A_T  = \frac{1}{m}\frac{\partial B^{h,\Bar{\lambda}}}{\partial \vect{p}}\trans\Hat{\boldsymbol{k}}_{\Bar{B}}$ and $\vect{A}_\tau  = \frac{\Bar{u}_T}{m}J^{-1}\left(\Hat{\boldsymbol{k}}\times \boldsymbol{o}^*\frac{\partial B^{h,\Bar{\lambda}}}{\partial \vect{p}}\boldsymbol{o}\right)$.}

\GR{Notice that under Assumptions~\ref{assumption:positive_uT} and \ref{assumption:smmothed_distance_vector_never_null}, the terms $A_T$ and $\vect{A}_\tau$ are well-defined.
However, pure barrier enforcement may lead to a trivial equilibrium where $\vect{v} = 0$, causing the vehicle to stop to avoid collision \cite{10472718}. To overcome this deadlock, a circulation component must be introduced when $F < 0$, following \cite{nunes2026safe, nunes2026saferobusttubebasedpathfollowing}. To design the circulation constraint, we construct an auxiliary HOCBF sequence. Let} \GR{$\phi_1(\vect{p},\vect{v},t) = M\frac{\partial B^{h,\Bar{\lambda}}}{\partial \vect{p}}\trans\vect{v} - b_1$, with $b_1 > 0$, and $\phi_i = \Dot{\phi}_{i-1} + \beta_i(\phi_{i-1})$}, 
where $M$ is a skew-symetric matrix \GR{chosen} according to \cite{nunes2026safe,nunes2026saferobusttubebasedpathfollowing}\GR{, and each} $\beta_i$ is a $\mathcal{K}$-class function. To ensure \GR{active obstacle} circulation \GR{during evasive maneuvers}, we \GR{enforce} $\phi_4(\vect{p},\vect{v},\vect{o},\vect{\omega}, \vect{u}_\tau, u_T, \Dot{u}_T, \Ddot{u}_T)\ge 0$ \GR{whenever} $F<0$. 

Let \GR{$E=\phi_4(\vect{p},\vect{v},\vect{o},\vect{\omega}, \vect{\tau}_\Phi, u_T, \Dot{u}_T, T_{2,\Phi})$} \GR{denote} the circulation HOCBF evaluated at the \GR{nominal} path-following control inputs $\vect{\tau}_\Phi$ and $T_{2,\Phi}$, and let \GR{$E_\Psi=T_{2,\Psi}B_T + \vect{\tau}_\Psi\trans\vect{B}_\tau$} \GR{represent the contribution resulting from the safety correction inputs} $\vect{\tau}_\Psi$ and $T_{2,\Psi}$, 
where $B_T = \frac{1}{m}M\frac{\partial B^{h,\Bar{\lambda}}}{\partial \vect{p}}\trans\Hat{\boldsymbol{k}}_{\Bar{B}}$ and $\vect{B}_\tau = \frac{\Bar{u}_T}{m}J^{-1}\left(\Hat{\boldsymbol{k}}\times \boldsymbol{o}^*M\frac{\partial B^{h,\Bar{\lambda}}}{\partial \vect{p}}\boldsymbol{o}\right)$.

%
%
%

\GR{If $E + E_\Psi \ge 0$, the safety-corrected control inputs simultaneously preserve safety and maintain forward motion. Conversely, if $E + E_\Psi < 0$, the control laws must be updated to jointly satisfy both the barrier constraint ($\psi_4 \ge 0$) and the circulation constraint ($\phi_4 \ge 0$). The joint control corrections $(\boldsymbol{\tau}_\Theta, T_{2,\Theta})$ are obtained as
\begin{align}
  T_{2,\Theta} & = aA_T + bB_T, \quad \vect{\tau}_\Theta = a\vect{A}_\tau + b\vect{B}_\tau,\label{eq:u_theta}
\end{align}
where the Gram matrix elements $g_{ij}$ associated with the control mapping directions are defined by $g_{11}  = A_T^2 + \|\vect{A}_\tau\|^2$, $g_{22}  = B_T^2 + \|\vect{B}_\tau\|^2$, and $g_{12} = A_TB_T + \vect{A}_\tau\trans\vect{B}_\tau$, while the dual multipliers $a$ and $b$ are given explicitly by
\begin{align}
    a & = \frac{-Fg_{22} + Eg_{12}}{g_{11}g_{22} - g_{12}^2}, \quad b = \frac{-Eg_{11} + Fg_{12}}{g_{11}g_{22} - g_{12}^2}.
\end{align}}

\GR{Finally, the complete closed-loop torque and thrust-rate control laws are obtained via piecewise projection}
\begin{align}
    \vect{u}_\tau & = \vect{\tau}_\Phi + 
\begin{cases}
\vect{0} & F \ge 0\\
\vect{\tau}_\Psi & F < 0 \ \text{and} \ E+E_\Psi\ge 0,\\
\vect{\tau}_\Theta & \text{otherwise},
\end{cases}\label{eq:control_law_u_tau}\\
\Ddot{u}_T & = T_{2,\Phi} + 
\begin{cases}
0 & F \ge 0,\\
T_{2,\Psi} & F < 0 \ \text{and} \ E+E_\Psi\ge 0,\\
T_{2,\Theta} & \text{otherwise}.
\end{cases}\label{eq:control_law_ddot_u_T}
\end{align}

\begin{lemma}[Safety]
     \GR{Consider the nominal system dynamics} \eqref{eq:nominal-model1}-\eqref{eq:nominal-model4} in initial states such that $\psi_i(x_0,0) \ge 0$ for $i \in [1,2,3,4]$. If Assumptions \ref{assumption:positive_uT} and \ref{assumption:smmothed_distance_vector_never_null} hold and the control laws $\Bar{\vect{u}}_\tau$ from \eqref{eq:control_law_u_tau} and $\Ddot{u}_T$ from \eqref{eq:control_law_ddot_u_T} are applied, then the safety set $\vect{\mathcal{S}}(t)$ is forward invariant. 
     \label{lemma:safety}
\end{lemma}
\begin{proof}
    When \GR{applying} the control laws $\Bar{\vect{u}}_\tau$ from \eqref{eq:control_law_u_tau} and $\Ddot{u}_T$ from \eqref{eq:control_law_ddot_u_T}, \GR{we evaluate the fourth-order barrier function} $\psi_4(\vect{p},\vect{v},\vect{o},\vect{\omega}, \vect{u}_\tau, u_T, \Dot{u}_T, \Ddot{u}_T)$ \GR{across three cases}. For the first case, when $F\ge0$, $\psi_4 \ge 0$ by the definition of $F$ \eqref{eq:definition_F}. When $F < 0$ and $E+E_\Psi\ge 0$\GR{, $\psi_4 = F + T_{2,\Psi}A_T + \vect{\tau}_\Psi\trans\vect{A}_\tau$}. 
    Substituting \eqref{eq:u_psi}, yields $\psi_4 = F - F = 0$.
    \GR{In} the last case, using \eqref{eq:u_theta}, $\psi_4$ takes the form $\psi_4 = F + T_{2,\Theta} A_T + \vect{\tau}_{\Theta} \cdot \vect{A}_\tau = 0$.
    
    \GR{Thus, $\psi_4 \ge 0$ holds for all conditions.} This result means that $\psi_0 = B^{h,\Bar{\lambda}}(\bvect{p},t) \ge 0$. 
    Finally, \GR{applying \cite[Lemma~2]{nunes2026safe}}, $B^{h,\Bar{\lambda}}(\bvect{p},t) \ge 0 \implies B^{\Bar{\lambda}}(\bvect{p},t) \ge 0 $\GR{, which guarantees that} the system \GR{trajectories} remain inside $\vect{\mathcal{S}}(t)$.
\end{proof}

\begin{theorem}
     \GR{Consider the nominal quadrotor dynamics} \eqref{eq:nominal-model1}-\eqref{eq:nominal-model4} under Assumptions \ref{assumption:positive_uT} and \ref{assumption:smmothed_distance_vector_never_null}. If the control laws $\Bar{\vect{u}}_\tau$ from \eqref{eq:control_law_u_tau} and $\Ddot{u}_T$ from \eqref{eq:control_law_ddot_u_T} \GR{are applied, then the safety set $\vect{\mathcal{S}}(t)$ is forward invariant, 
     solving Problem~\ref{problem:collision_avoidance}.} Furthermore, \GR{when unconstrained by the barrier condition (i.e., when $F \ge 0$), the system trajectories asymptotically converge to and follow the target path $\vect{\mathcal{C}}(t)$, solving Problem~\ref{problem:path_following2}.} 
     \label{theorem:main}
\end{theorem}
\begin{proof}
    Problem \ref{problem:collision_avoidance} is solved as a result of Lemma \ref{lemma:safety}. When $F\ge0$, Problem \ref{problem:path_following2} is solved as a result of Theorem \ref{lemma:path_following}.
\end{proof}

\subsection{Vehicle yaw angle \GR{control}}
As \GR{noted in Section~\ref{subsec:backstepping}}, $\omega_{yaw}$ \GR{acts as} a free parameter \GR{in} the control law \eqref{eq:varphi_control_law} \GR{and} does not affect the \GR{convergence guarantees} of Theorem \ref{theorem:main}. \GR{This degree of freedom allows for secondary control objectives, such as controlling the vehicle yaw angle independently.}


\GR{A common requirement is to align the vehicle yaw angle with a specified direction $\boldsymbol{x} \triangleq x_i \hat{\boldsymbol{i}} + x_j \hat{\boldsymbol{j}}$ in the horizontal plane, where $\|\boldsymbol{x}\| > 0$, with time derivative $\dot{\boldsymbol{x}}$. Let $\psi(\bar{\boldsymbol{o}}) \in \mathbb{S}^1$ denote the current yaw angle extracted from the quaternion $\bar{\boldsymbol{o}}$.} \GR{A proportional-plus-feedforward yaw rate control law} $\omega_{yaw}$ is given by
\begin{align}
    z_{yaw} = & -\mathrm{atan2}(\cos(\psi)x_j - \sin(\psi)x_i,\nonumber\\
    &\cos(\psi)x_i + \sin(\psi)x_j)\\
    f_f = & \frac{x_i\Dot{x}_j - x_j\Dot{x}_i}{\|\boldsymbol{x}\|}\\
    \omega_{yaw} = & -k_{yaw} z_{yaw} + f_f, \quad k_{yaw} > 0.
\end{align}

\section{Results}

\GR{To provide a} realistic validation, we implemented the \GR{proposed} controller \GR{directly into} the firmware \GR{of} the Crazyflie nano-quadcopter \GR{platform}\footnote{https://www.bitcraze.io/} and \GR{conducted} SITL simulations using Crazysim \cite{crazysim}, Gazebo, and Crazyswarm 2 \cite{crazyswarm}. \GR{Using the same firmware implementation evaluated in the SITL, the controller code was also compiled and flashed onto a physical Crazyflie 2.1 nano-quadcopter to conduct real-world flight experiments.} 


Detailed visualizations and animated videos of the simulated and experimented trajectories are available at \href{https://youtu.be/kcnU7gSoNiE}{youtu.be/kcnU7gSoNiE}. The implementation code for this project is available at \href{https://github.com/ArthurHDN/crazyflie-control-workspace}{github.com/ArthurHDN/crazyflie-control-workspace}.

\subsection{Crazysim SITL}

First, we evaluated our controller against four baseline controllers available in the Crazyflie firmware: \GR{a PID controller}, \GR{the} Mellinger \GR{controller} \cite{5980409} (MELL), \GR{the} Brescianini \GR{controller} \cite{BrescianiniNonlinearController2013} (BRES), and \GR{the Incremental Nonlinear Dynamic Inversion (INDI) controller} \cite{Smeur2016}. \GR{To} enable \GR{a fair} comparison, all controllers were subjected to position regulation tasks \GR{comprising} a sequence of \GR{step} position references. In our controller, position regulation is \GR{achieved simply} by \GR{setting} the tangent vector $\vect{T}_{\vect{\mathcal{C}}}=0$ \GR{in} \eqref{eq:field_T}\GR{, effectively disabling path progress}. 

\GR{Fig.} \ref{fig:qquad-crazysim-comparison} \GR{illustrates} the position \GR{step responses} for each controller. The INDI controller \GR{became unstable and} crashed \GR{early} during the simulation. The MELL controller \GR{exhibited aggressive behavior and successfully followed initial setpoints, but suffered instability} near $t=45\mathrm{s}$. The PID controller completed \GR{the sequence but exhibited significant} overshoot, \GR{momentarily losing altitude during rapid descents before recovering to the setpoint}. The BRES controller \GR{fulfilled} all \GR{regulation} tasks with little overshoot. \GR{Finally, our} controller \GR{successfully completed} all tasks \GR{with a smooth, well-damped response and a} prescribed convergence speed \GR{parameterized by} $v_r$. 

These results indicate that our controller maintains stable and damped behavior even when performing relatively fast position changes. In contrast to the PID and Mellinger controllers, which exhibited noticeable overshoot or instability during the sequence, ours achieved smooth convergence without significant transient oscillations. This is particularly relevant to the safe path-following application, since aggressive transients during path maneuvers could otherwise produce unnecessarily large tracking errors and increase the likelihood of approaching obstacles.

\GR{Note that} path-following and obstacle avoidance \GR{maneuvers} were not \GR{evaluated in this comparison,} as our \GR{controller} is the only \GR{one among those tested} that explicitly accounts for them. 
\begin{figure}[htb!]
	\centering
	\includegraphics[clip,trim={0in 0in 0in 0in}, width=0.99\columnwidth]{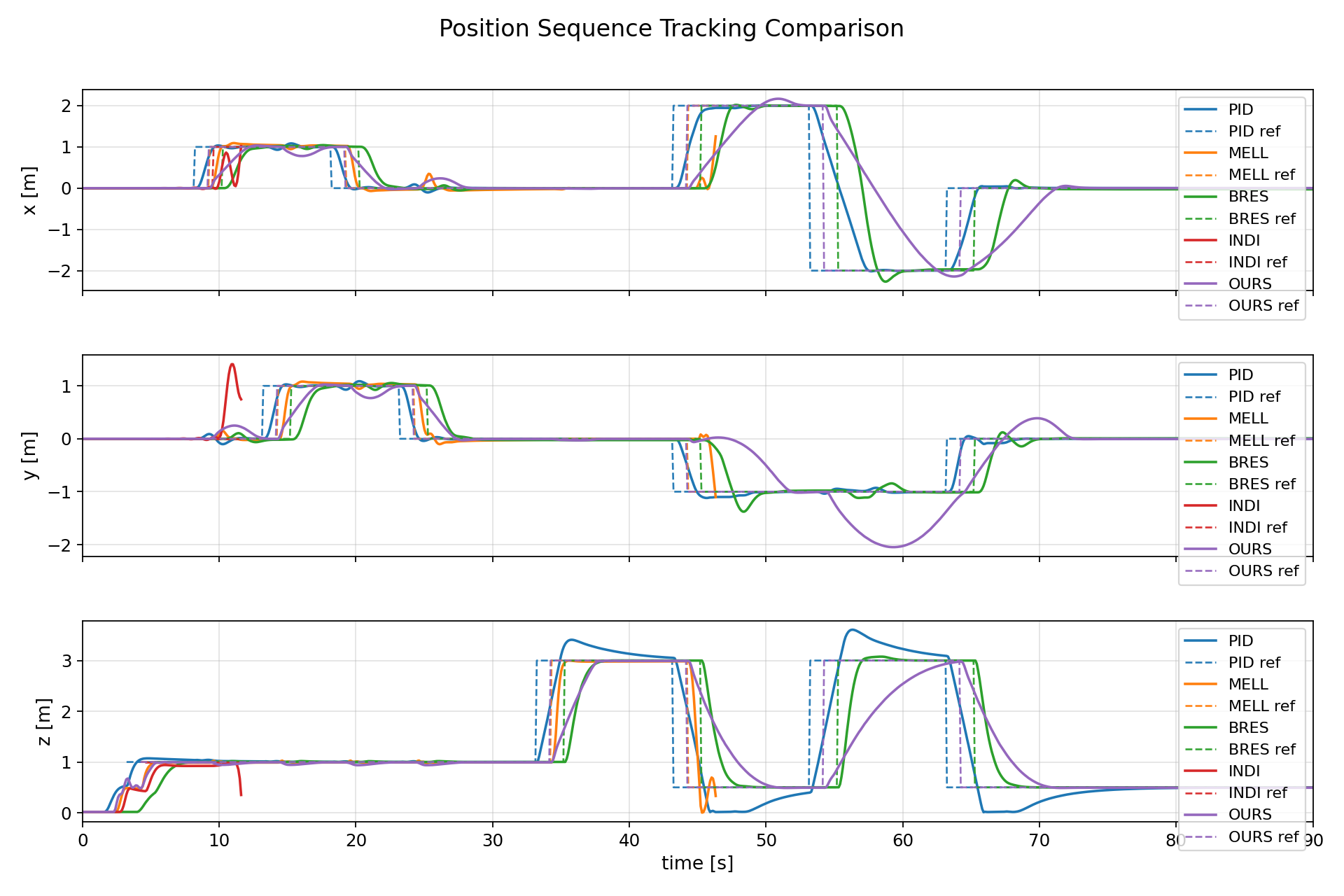}
	\caption{Comparison of our controller (OURS) with PID, Mellinger \cite{5980409} (MELL), Brescianini \cite{BrescianiniNonlinearController2013} (BRES), and INDI \cite{Smeur2016} for sequence of position regulation.}
	\label{fig:qquad-crazysim-comparison}
\end{figure}

In the second SITL simulation scenario, the quadcopter was commanded to follow a 3D parametric curve defined by
\begin{align}
    \vect{c}(s) = \begin{bmatrix}
        (0.25 + 0.1\cos(3s))\cos(2s) \\
        (0.25 + 0.1\cos(3s))\sin(2s) \\
        0.8 + 0.15\sin(3s)
    \end{bmatrix}.
\end{align}
\GR{As shown in Fig.~\ref{fig:qquad-crazysim-pathfollowing}, the proposed controller successfully guides the vehicle to converge to and follow the target trajectory with high accuracy.} This result illustrates the expected behavior of the AVF guidance. The vehicle initially approaches the path from a initial error and progressively reduces it before transitioning to a circulating motion along the path. The smooth decay of the distance to path metric indicates that the translational guidance and whole-body Backstepping controller remain consistent during this transition, without introducing significant oscillations around the desired path.
\begin{figure}[htb!]
	\centering
	\includegraphics[clip,trim={0in 0in 0in 0in}, width=0.55\columnwidth]{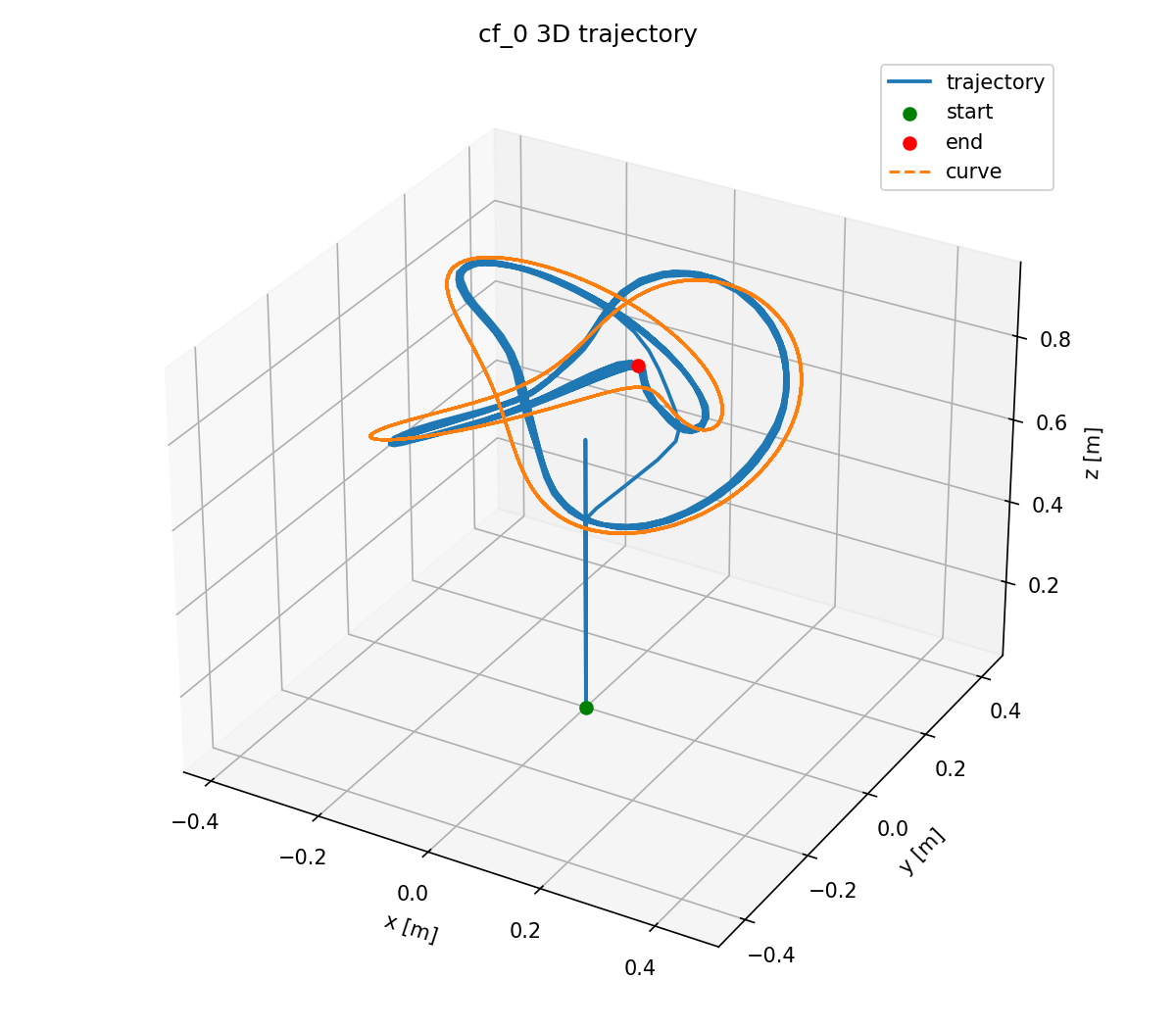}
    \includegraphics[clip,trim={0in 3.45in 0in 0in}, width=0.9\columnwidth]{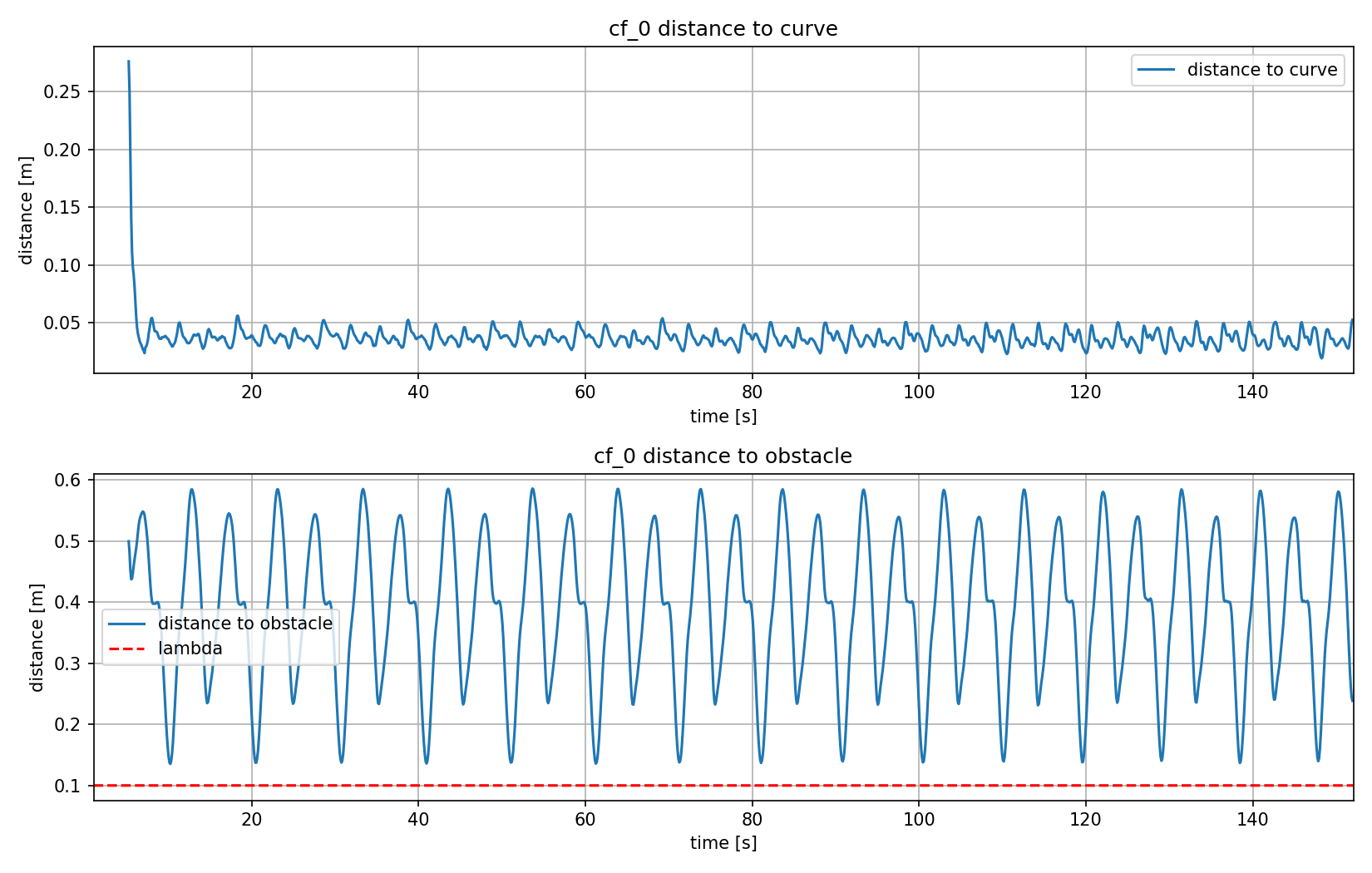}
	\caption{System trajectory and distance to the path of the Crazysim path-following task.}
	\label{fig:qquad-crazysim-pathfollowing}
\end{figure}

\GR{Finally, we evaluated the combined path-following and safe control framework in an obstacle avoidance scenario. The quadrotor was commanded to follow a 3D lemniscate curve}
\begin{align}
    \vect{c}(s) = \begin{bmatrix}
        0.35\sin(s) &
        0.25\sin(2s) &
        0.8 + 0.1\cos(s)
    \end{bmatrix}^{\top},
    \label{eq:path_3d_lemminiscata}
\end{align}
\GR{while avoiding a virtual obstacle represented by a point cloud. This demonstrate the interaction between the nominal path-following and the safety layer. As shown in Fig.~\ref{fig:qquad-crazysim-obstacle}, the proposed controller successfully maintains safety by steering the vehicle away from the obstacle when necessary and seamlessly returning it to the target path once the obstacle is cleared.}
\begin{figure}[htb!]
	\centering
	\includegraphics[clip,trim={0in 0in 0in 0in}, width=0.55\columnwidth]{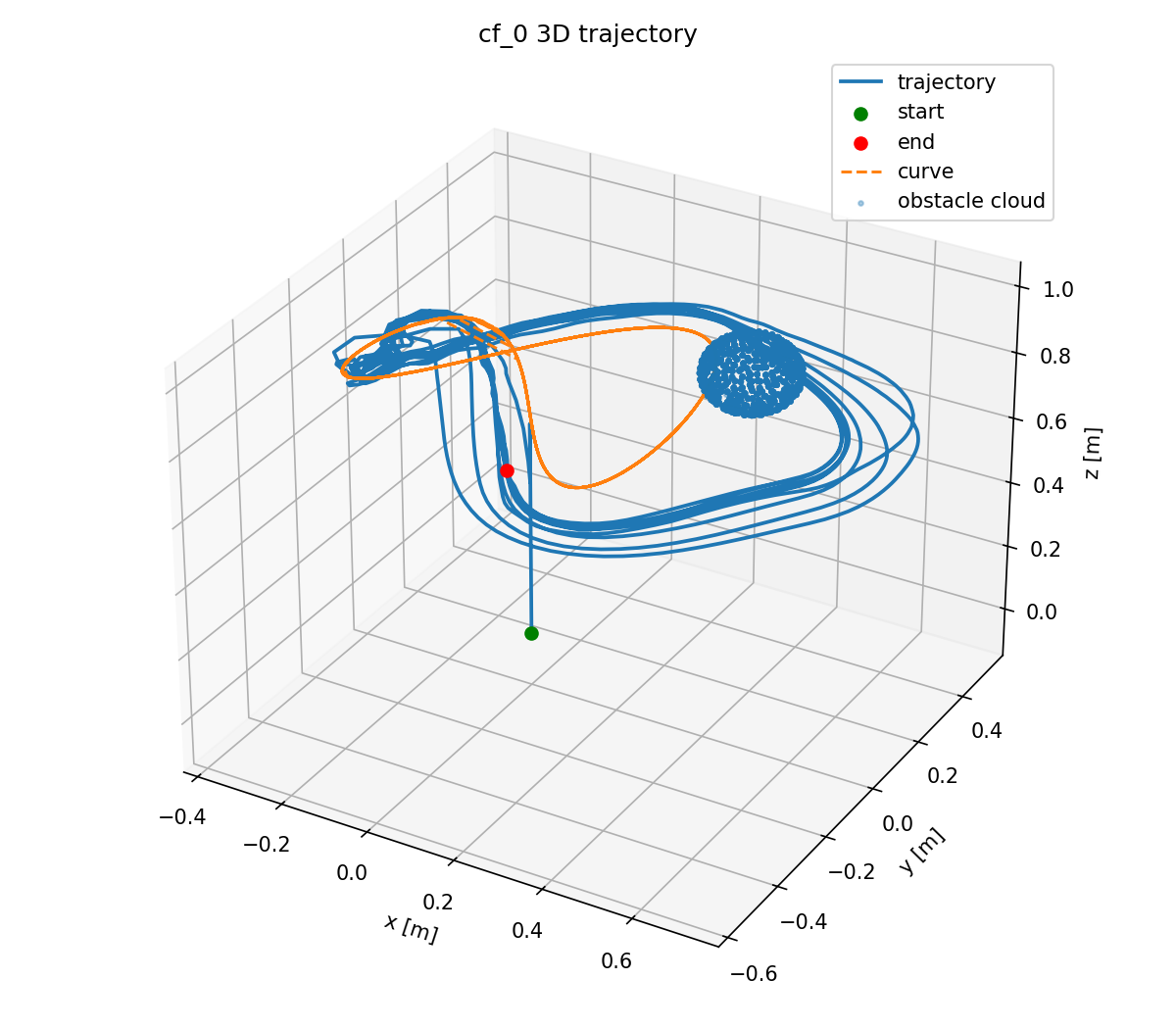}
    \includegraphics[clip,trim={0in 0in 0in 0in}, width=0.9\columnwidth]{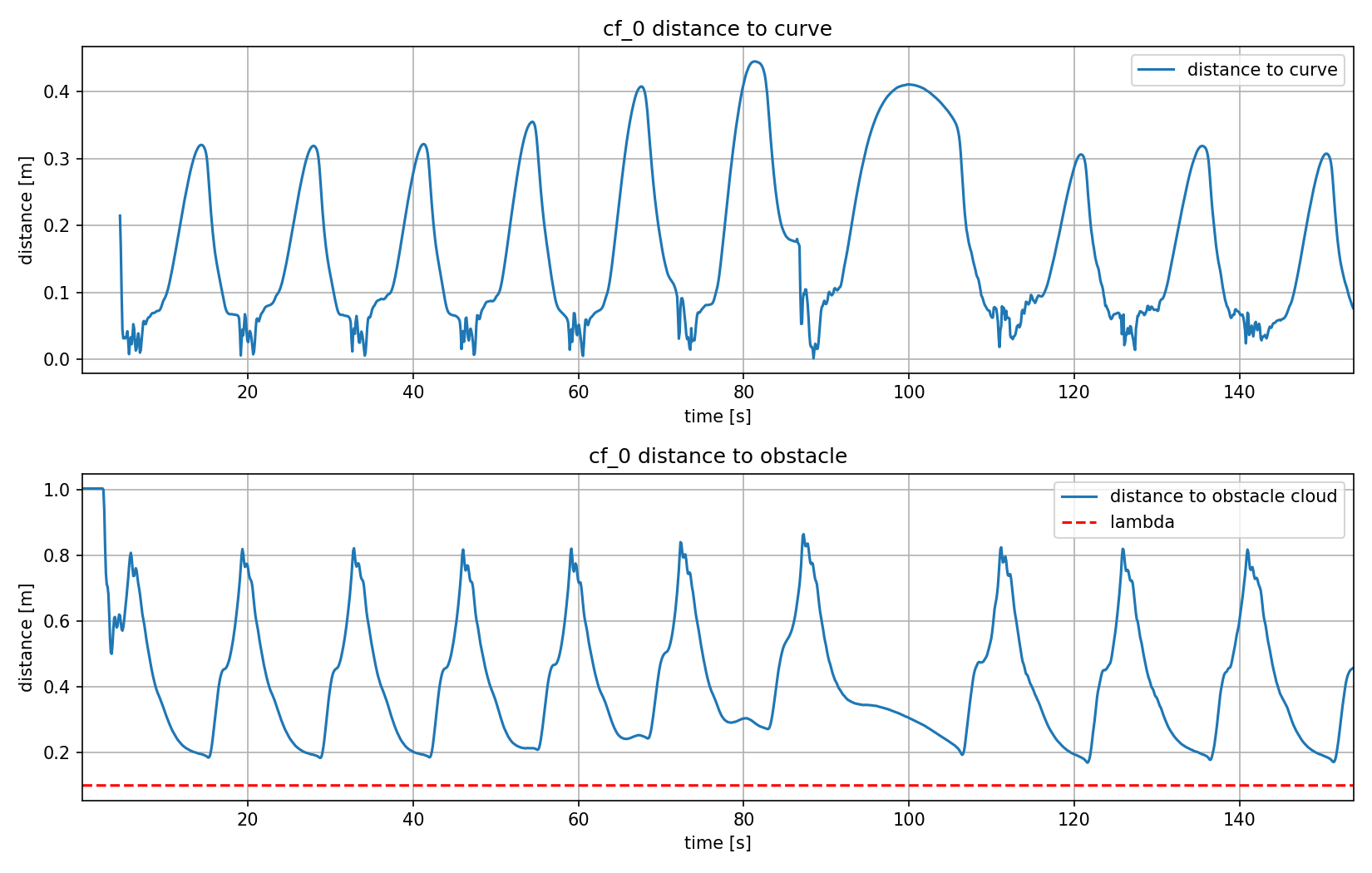}
	\caption{System trajectory and distance to the path of the Crazysim path-following with obstacle avoidance task.}
	\label{fig:qquad-crazysim-obstacle}
\end{figure}

\subsection{\GR{Experimental Results with the Crazyflie Nano-Quadcopter}}

\GR{The vehicle was commanded to follow the 3D lemniscate path defined in \eqref{eq:path_3d_lemminiscata} under two distinct operational conditions: unconstrained nominal path following, and path following in the presence of an obstacle. Figs.~\ref{fig:qquad-crazyflie-pathfollowing} and \ref{fig:qquad-crazyflie-obstacle} present the experimental 3D flight trajectories alongside the corresponding distance metrics for the obstacle-free and obstacle-avoidance scenarios, respectively.} Despite the considerably more challenging real-world conditions, including measurement noise, disturbances, and modeling discrepancies, the vehicle remains capable of following the prescribed path. The similarity between the simulated and experimental trajectories suggests that the controller retains its qualitative behavior when implemented on the physical platform.
\begin{figure}[htb!]
	\centering
	\includegraphics[clip,trim={0in 0in 0in 0in}, width=0.55\columnwidth]{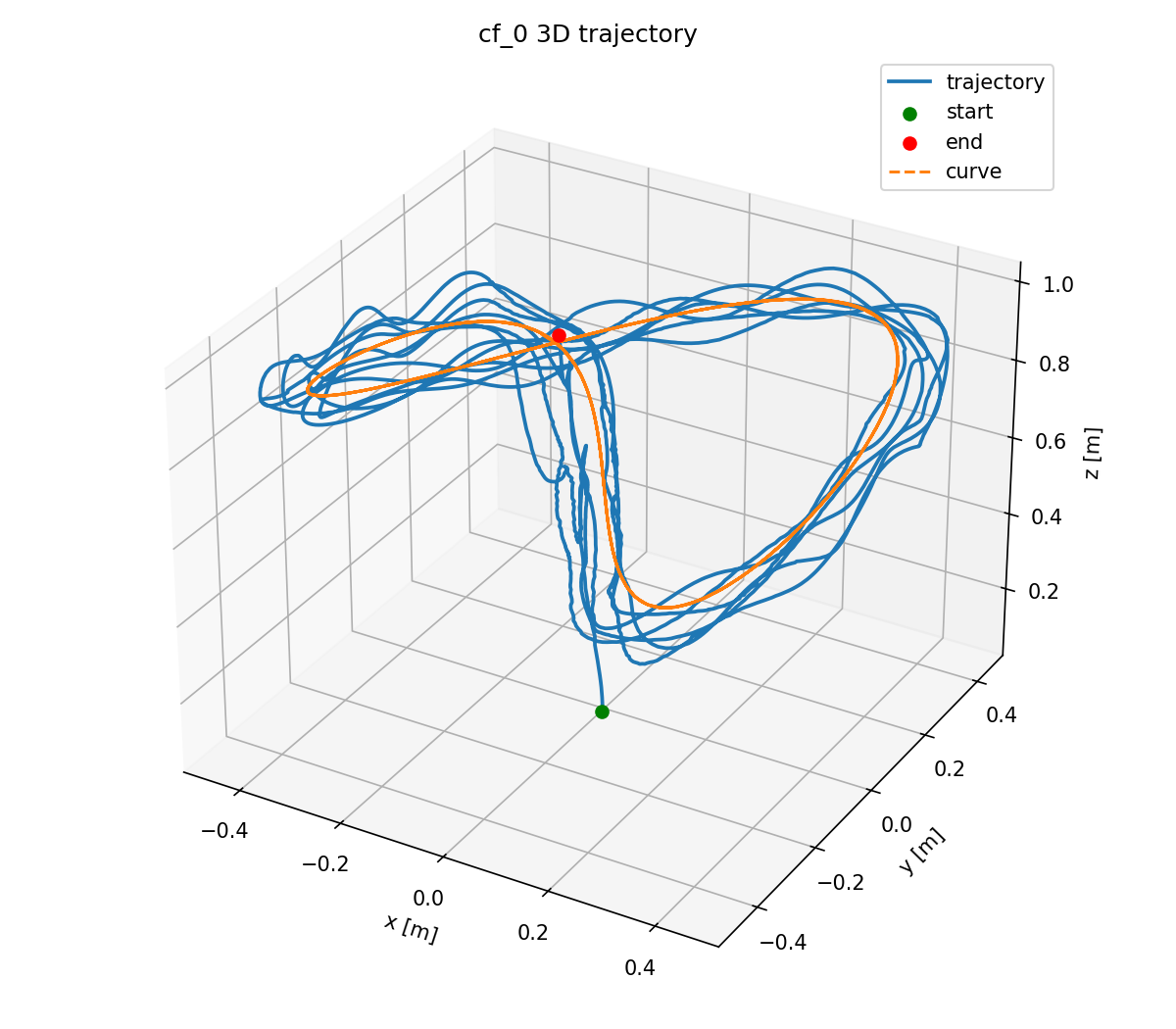}
    \includegraphics[clip,trim={0in 3.7in 0in 0in}, width=0.9\columnwidth]{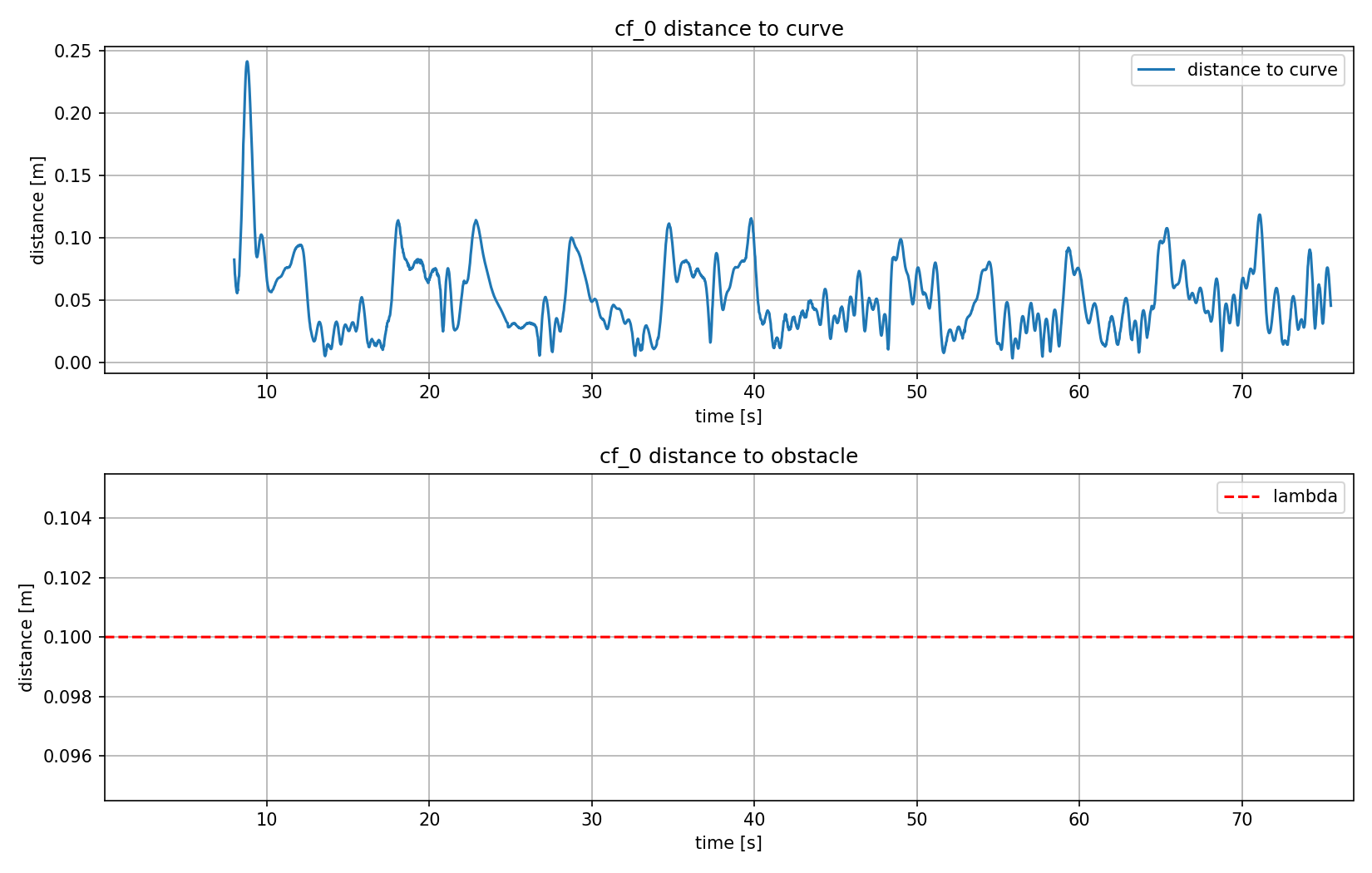}
	\caption{System trajectory and distance to the path of the Crazyflie path-following task experiment.}
	\label{fig:qquad-crazyflie-pathfollowing}
\end{figure}
\begin{figure}[htb!]
	\centering
	\includegraphics[clip,trim={0in 0in 0in 0in}, width=0.55\columnwidth]{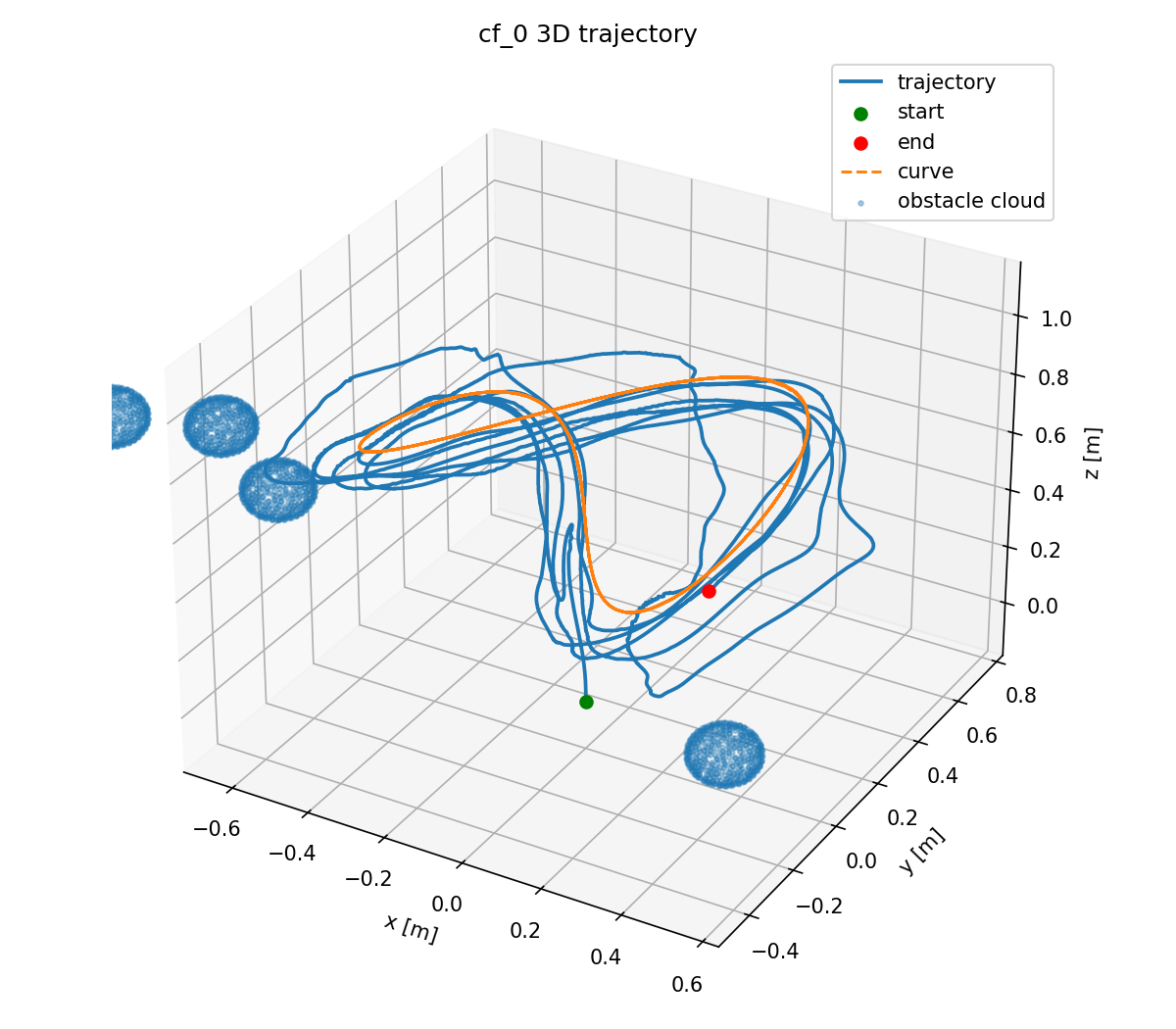}
    \includegraphics[clip,trim={0in 0in 0in 0in}, width=0.9\columnwidth]{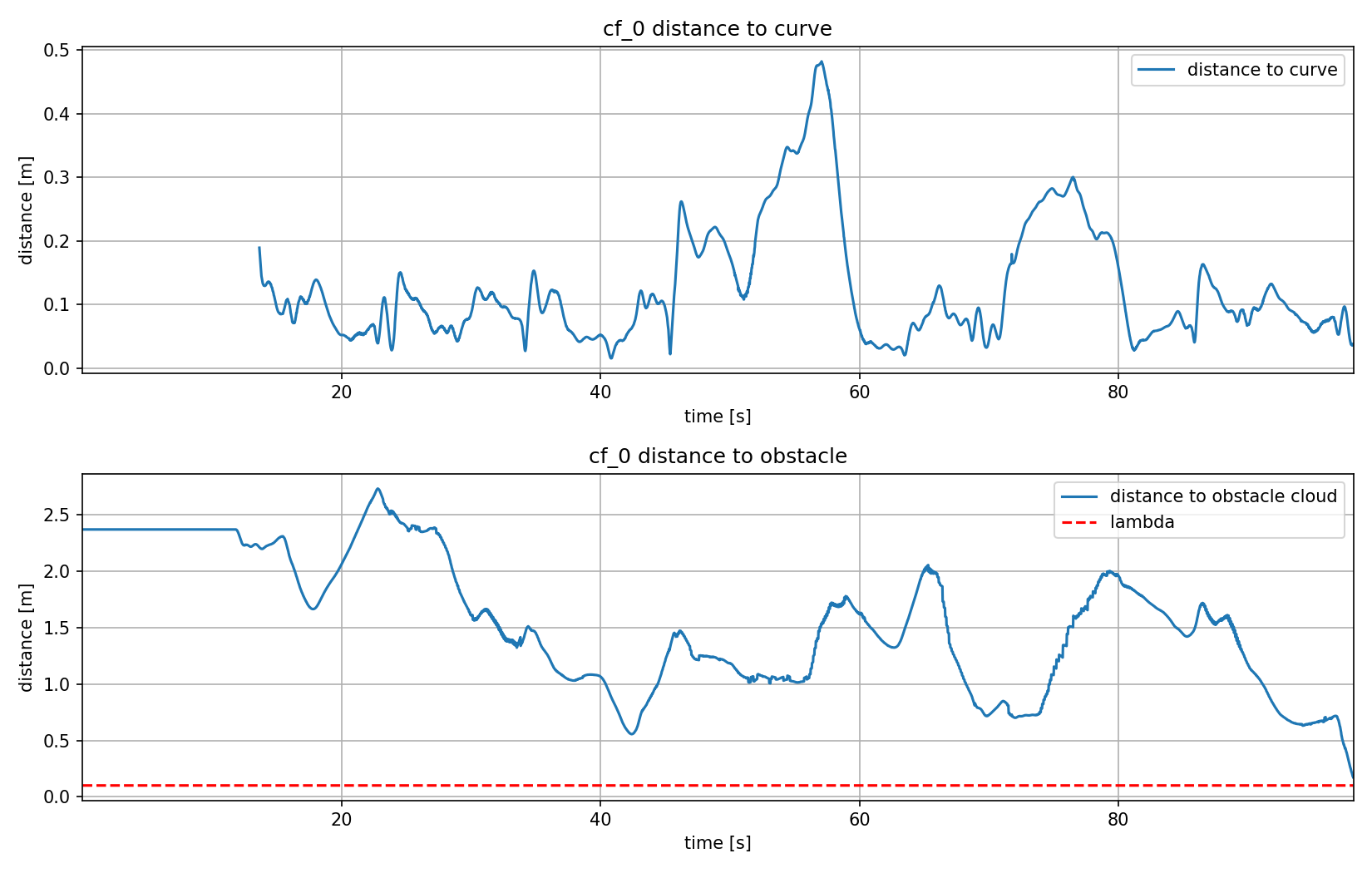}
	\caption{System trajectory and distance to the path of the Crazyflie path-following with obstacle avoidance task.}
	\label{fig:qquad-crazyflie-obstacle}
\end{figure}

\clearpage
\section{Conclusion}

\GR{This paper presented an integrated safety-critical control framework for quadcopter path following in the presence of obstacles. By unifying a vector-field navigation law with a whole-body backstepping controller and High-Order Control Barrier Functions (HOCBFs), the proposed scheme ensures high-accuracy tracking of time-varying 3D target curves while guaranteeing safety with respect to the full relative-degree-four nonlinear quadcopter dynamics. To overcome potential stagnation at singular barrier equilibria, an explicit closed-form dual-HOCBF circulation mechanism was incorporated to seamlessly detour around obstacles without bringing the vehicle to a stop. Rigorous theoretical analyses established the asymptotic convergence to the target path when unconstrained and the forward invariance of the safe set under all operating conditions. Finally, the controller was implemented directly in the Crazyflie firmware and validated through both SITL simulations and real-world flight experiments, demonstrating robust obstacle avoidance and damped tracking performance.}



\bibliographystyle{ieeetr} 
\bibliography{bibliography}

@article{raffo2010integral,
  title={An integral predictive/nonlinear H$_\infty$ control structure for a quadrotor helicopter},
  author={Raffo, Guilherme V and Ortega, Manuel G and Rubio, Francisco R},
  journal={Automatica},
  volume={46},
  number={1},
  pages={29--39},
  year={2010},
  publisher={Elsevier}
}

@article{khalil2002nonlinear,
  title={Nonlinear systems},
  author={Khalil, Hassan K},
  journal={Patience Hall},
  volume={115},
  year={2002}
}

@article{rezende2021constructive,
  author={Rezende, Adriano M. C. and Goncalves, Vinicius M. and Pimenta, Luciano C. A.},
  journal={IEEE Transactions on Robotics}, 
  title={Constructive Time-Varying Vector Fields for Robot Navigation}, 
  year={2021},
  volume={38},
  number={2},
  pages={852-867},
  doi={10.1109/TRO.2021.3093674}}

@inproceedings{rezende2020robust,
  title={Robust quadcopter control with artificial vector fields},
  author={Rezende, Adriano MC and Gon{\c{c}}alves, Vinicius M and Nunes, Arthur HD and Pimenta, Luciano CA},
  booktitle={2020 IEEE International Conference on Robotics and Automation (ICRA)},
  pages={6381--6387},
  year={2020},
  organization={IEEE}
}

@unpublished{adorno2017robot,
  TITLE = {{Robot Kinematic Modeling and Control Based on Dual Quaternion Algebra --- Part I: Fundamentals.}},
  AUTHOR = {Vilhena Adorno, Bruno},
  URL = {https://hal.science/hal-01478225},
  YEAR = {2017},
  MONTH = Feb,
  HAL_ID = {hal-01478225},
  HAL_VERSION = {v1},
}

@inproceedings{pereira2021collision,
  title={Collision-free vector field guidance and {MPC} for a fixed-wing {UAV}},
  author={Pereira, Leonardo AA and Nunes, Arthur HD and Rezende, Adriano MC and Gon{\c{c}}alves, Vinicius M and Raffo, Guilherme V and Pimenta, Luciano CA},
  booktitle={2021 IEEE International Conference on Robotics and Automation (ICRA)},
  pages={176--182},
  year={2021},
  organization={IEEE}
}

@article{rubi2020survey,
  title={A survey of path following control strategies for {UAV}s focused on quadrotors},
  author={Rub{\'\i}, Bartomeu and P{\'e}rez, Ramon and Morcego, Bernardo},
  journal={Journal of Intelligent \& Robotic Systems},
  volume={98},
  number={2},
  pages={241--265},
  year={2020},
  publisher={Springer}
}

@inproceedings{nunes2022vector,
  author={Nunes, Arthur H. D. and Rezende, Adriano M. C. and Cruz, Gilmar P. and Freitas, Gustavo M. and Gonçalves, Vinicius M. and Pimenta, Luciano C. A.},
  booktitle={2022 IEEE 61st Conference on Decision and Control (CDC)}, 
  title={Vector field for curve tracking with obstacle avoidance}, 
  year={2022},
  volume={},
  number={},
  pages={2031-2038},
  doi={10.1109/CDC51059.2022.9992435}}

@mastersthesis{salierno2018whole,
  author  = "Giovanni Felice Salierno",
  title   = "Whole-body Backstepping Control of a Quadrotor {UAV}  for Trajectory Tracking",
  school  = "Universidade Federal de Minas Gerais",
  year    = "2018"
}

@inproceedings{8714418,
  author={Eid, Saif Eldin and Sham Dol, Sharul},
  booktitle={2019 Advances in Science and Engineering Technology International Conferences (ASET)}, 
  title={Design and Development of Lightweight-High Endurance Unmanned Aerial Vehicle for Offshore Search and Rescue Operation}, 
  year={2019},
  volume={},
  number={},
  pages={1-5},
  doi={10.1109/ICASET.2019.8714418}}

@article{miranda2022autonomous,
  title={Autonomous navigation system for a delivery drone},
  author={Miranda, Victor RF and Rezende, Adriano and Rocha, Thiago L and Azp{\'u}rua, H{\'e}ctor and Pimenta, Luciano CA and Freitas, Gustavo M},
  journal={Journal of Control, Automation and Electrical Systems},
  volume={33},
  number={1},
  pages={141--155},
  year={2022},
  publisher={Springer}
}

@article{silva2022dynamics,
  title={Dynamics of mobile manipulators using dual quaternion algebra},
  author={Afonso Silva, Frederico Fernandes and Jos{\'e} Quiroz-Oma{\~n}a, Juan and Vilhena Adorno, Bruno},
  journal={Journal of Mechanisms and Robotics},
  volume={14},
  number={6},
  pages={061005},
  year={2022},
  publisher={American Society of Mechanical Engineers}
}

@article{yao2021singularity,
  year = {2021},
  month = aug,
  publisher = {Institute of Electrical and Electronics Engineers ({IEEE})},
  volume = {37},
  number = {4},
  pages = {1206--1221},
  author = {Weijia Yao and Hector Garcia de Marina and Bohuan Lin and Ming Cao},
  title = {{Singularity-Free Guiding Vector Field for Robot Navigation}},
  journal = {{IEEE} Transactions on Robotics}
}

@inproceedings{8384677,
  author={Williams, Alexander and Yakimenko, Oleg},
  booktitle={2018 4th International Conference on Control, Automation and Robotics (ICCAR)}, 
  title={Persistent mobile aerial surveillance platform using intelligent battery health management and drone swapping}, 
  year={2018},
  volume={},
  number={},
  pages={237-246},
  doi={10.1109/ICCAR.2018.8384677}}

@inproceedings{7487716,
  author={Duggal, Vishakh and Sukhwani, Mohak and Bipin, Kumar and Reddy, G. Syamasundar and Krishna, K. Madhava},
  booktitle={2016 IEEE International Conference on Robotics and Automation (ICRA)}, 
  title={Plantation monitoring and yield estimation using autonomous quadcopter for precision agriculture}, 
  year={2016},
  volume={},
  number={},
  pages={5121-5127},
  doi={10.1109/ICRA.2016.7487716}}

@inproceedings{5531424,
  author={Huang, Mu and Xian, Bin and Diao, Chen and Yang, Kaiyan and Feng, Yu},
  booktitle={Proceedings of the 2010 American Control Conference}, 
  title={Adaptive tracking control of underactuated quadrotor unmanned aerial vehicles via backstepping}, 
  year={2010},
  volume={},
  number={},
  pages={2076-2081},
  doi={10.1109/ACC.2010.5531424}}

@article{5428769,
  author={Raptis, Ioannis A. and Valavanis, Kimon P. and Moreno, Wilfrido A.},
  journal={IEEE Transactions on Control Systems Technology}, 
  title={A Novel Nonlinear Backstepping Controller Design for Helicopters Using the Rotation Matrix}, 
  year={2011},
  volume={19},
  number={2},
  pages={465-473},
  doi={10.1109/TCST.2010.2042450}}

@article{TAN201647,
title = {Tracking of a moving ground target by a quadrotor using a backstepping approach based on a full state cascaded dynamics},
journal = {Applied Soft Computing},
volume = {47},
pages = {47-62},
year = {2016},
issn = {1568-4946},
author = {Chun Kiat Tan and Jianliang Wang and Yew Chai Paw and Teng Yong Ng}
}

@article{abaunza2017dual,
  title={Dual quaternion modeling and control of a quad-rotor aerial manipulator},
  author={Abaunza, Hern{\'a}n and Castillo, Pedro and Victorino, A and Lozano, Rogelio},
  journal={Journal of Intelligent \& Robotic Systems},
  volume={88},
  number={2},
  pages={267--283},
  year={2017},
  publisher={Springer}
}

@inproceedings{5980409,
author={D. Mellinger and V. Kumar},
booktitle={2011 IEEE International Conference on Robotics and Automation},
title={Minimum snap trajectory generation and control for quadrotors},
year={2011},
volume={},
number={},
pages={2520-2525},
doi={10.1109/ICRA.2011.5980409},
ISSN={1050-4729},
month={May},}

@article{yao2022guiding,
  author={Yao, Weijia and Lin, Bohuan and Anderson, Brian D. O. and Cao, Ming},
  journal={IEEE Transactions on Automatic Control}, 
  title={Guiding Vector Fields for Following Occluded Paths}, 
  year={2022},
  volume={67},
  number={8},
  pages={4091-4106},
  doi={10.1109/TAC.2022.3179215}}

@article{ozbek2016feedback,
author = {Necdet Sinan Özbek and Mert Önkol and Mehmet Önder Efe},
title ={Feedback control strategies for quadrotor-type aerial robots: a survey},
journal = {Transactions of the Institute of Measurement and Control},
volume = {38},
number = {5},
pages = {529-554},
year = {2016},
doi = {10.1177/0142331215608427},

URL = { 
        https://doi.org/10.1177/0142331215608427
    
},
eprint = { 
        https://doi.org/10.1177/0142331215608427
    
}
}

@inproceedings{bouabdallah2004pid,
  author={Bouabdallah, S. and Noth, A. and Siegwart, R.},
  booktitle={2004 IEEE/RSJ International Conference on Intelligent Robots and Systems (IROS) (IEEE Cat. No.04CH37566)}, 
  title={{PID} vs {LQ} control techniques applied to an indoor micro quadrotor}, 
  year={2004},
  volume={3},
  number={},
  pages={2451-2456 vol.3},
  doi={10.1109/IROS.2004.1389776}}

@inproceedings{cowling2007prototype,
  author={Cowling, Ian D. and Yakimenko, Oleg A. and Whidborne, James F. and Cooke, Alastair K.},
  booktitle={2007 European Control Conference (ECC)}, 
  title={A prototype of an autonomous controller for a quadrotor {UAV}}, 
  year={2007},
  volume={},
  number={},
  pages={4001-4008},
  doi={10.23919/ECC.2007.7068316}}

@inproceedings{nunes2023integrated,
  author={Nunes, Arthur H. D. and Raffo, Guilherme V. and Pimenta, Luciano C. A.},
  booktitle={2023 IEEE International Conference on Robotics and Automation (ICRA)}, 
  title={Integrated vector field and backstepping control for quadcopters}, 
  year={2023},
  volume={},
  number={},
  pages={1256-1262},
  doi={10.1109/ICRA48891.2023.10160824}}

@INPROCEEDINGS{8796030,
  author={Ames, Aaron D. and Coogan, Samuel and Egerstedt, Magnus and Notomista, Gennaro and Sreenath, Koushil and Tabuada, Paulo},
  booktitle={2019 18th European Control Conference (ECC)}, 
  title={Control Barrier Functions: Theory and Applications}, 
  year={2019},
  volume={},
  number={},
  pages={3420-3431},
  doi={10.23919/ECC.2019.8796030}}

@ARTICLE{10530410,
  author={Gonçalves, Vinicius Mariano and Tzes, Anthony and Khorrami, Farshad and Fraisse, Philippe},
  journal={IEEE Transactions on Robotics}, 
  title={Smooth Distances for Second-Order Kinematic Robot Control}, 
  year={2024},
  volume={40},
  number={},
  pages={2950-2966},
  doi={10.1109/TRO.2024.3400924}}

@ARTICLE{10472718,
  author={Gonçalves, Vinicius Mariano and Krishnamurthy, Prashanth and Tzes, Anthony and Khorrami, Farshad},
  journal={IEEE Transactions on Control Systems Technology}, 
  title={Control Barrier Functions With Circulation Inequalities}, 
  year={2024},
  volume={32},
  number={4},
  pages={1426-1441},
  doi={10.1109/TCST.2024.3372802}}

@inproceedings{crazyswarm,
  author    = {James A. Preiss* and
               Wolfgang  H\"onig* and
               Gaurav S. Sukhatme and
               Nora Ayanian},
  title     = {Crazyswarm: {A} large nano-quadcopter swarm},
  booktitle = {{IEEE} International Conference on Robotics and Automation ({ICRA})},
  pages     = {3299--3304},
  publisher = {{IEEE}},
  year      = {2017},
  url       = {https://doi.org/10.1109/ICRA.2017.7989376},
  doi       = {10.1109/ICRA.2017.7989376},
}

@ARTICLE{nunes2026safe,
  author={Nunes, Arthur H. D. and Gonçalves, Vinicius M. and Pimenta, Luciano C. A.},
  journal={IEEE Robotics and Automation Letters}, 
  title={Safe Vector Field for Robot Navigation in $n$-Dimensions}, 
  year={2026},
  volume={11},
  number={3},
  pages={3071-3078},
  doi={10.1109/LRA.2026.3655306}}

@ARTICLE{11358638,
  author={Gonçalves, Vinicius M. and Wei, Shiqing and de Souza, Eduardo Malacarne Soeiro and Krishnamurthy, Prashanth and Tzes, Anthony and Khorrami, Farshad},
  journal={IEEE Robotics and Automation Letters}, 
  title={A Differentiable Distance Metric for Robotics Through Generalized Alternating Projection}, 
  year={2026},
  volume={11},
  number={3},
  pages={2889-2896},
  doi={10.1109/LRA.2026.3655272}}

@INPROCEEDINGS{crazysim,
  author={Llanes, Christian and Kakish, Zahi and Williams, Kyle and Coogan, Samuel},
  booktitle={2024 IEEE International Conference on Robotics and Automation (ICRA)}, 
  title={CrazySim: A Software-in-the-Loop Simulator for the Crazyflie Nano Quadrotor}, 
  year={2024},
  volume={},
  number={},
  pages={12248-12254},
  doi={10.1109/ICRA57147.2024.10610906}}

@ARTICLE{BrescianiniNonlinearController2013,
               title={Nonlinear quadrocopter attitude control},
               author={Brescianini, Dario and Hehn, Markus and D'Andrea, Raffaello},
               year={2013},
               publisher={ETH Zurich}}

@article{Smeur2016,
  title = {Adaptive Incremental Nonlinear Dynamic Inversion for Attitude Control of Micro Air Vehicles},
  volume = {39},
  ISSN = {1533-3884},
  url = {http://dx.doi.org/10.2514/1.G001490},
  DOI = {10.2514/1.g001490},
  number = {3},
  journal = {Journal of Guidance,  Control,  and Dynamics},
  publisher = {American Institute of Aeronautics and Astronautics (AIAA)},
  author = {Smeur,  Ewoud J. J. and Chu,  Qiping and de Croon,  Guido C. H. E.},
  year = {2016},
  month = Mar,
  pages = {450–461}
}

@INPROCEEDINGS{ferreira2026pointtocloudnmpcsmoothavoidance,
      title={Point-to-Cloud NMPC with Smooth Avoidance Constraints}, 
      author={Brener G. Ferreira and Vinicius M. Gonçalves and Marcelo A. Santos and Guilherme V. Raffo},
    booktitle={2026 European Control Conference (ECC)},
      year={2026},
      volume={},
      number={},
      pages={3346-3351},      
}

@article{raffo2015robust,
author = {Raffo, Guilherme V. and Ortega, Manuel G. and Rubio, Francisco R.},
title = {Robust Nonlinear Control for Path Tracking of a Quad-Rotor Helicopter},
journal = {Asian Journal of Control},
volume = {17},
number = {1},
pages = {142-156},
doi = {https://doi.org/10.1002/asjc.823},
url = {https://onlinelibrary.wiley.com/doi/abs/10.1002/asjc.823},
eprint = {https://onlinelibrary.wiley.com/doi/pdf/10.1002/asjc.823},
year = {2015}
}

@misc{nunes2026saferobusttubebasedpathfollowing,
      title={Safe and robust tube-based path-following for robot navigation}, 
      author={Arthur H. D. Nunes and Vinicius M. Gonçalves and Guilherme V. Raffo and Leonardo A. B. Torres and Luciano C. A. Pimenta},
      year={2026},
      eprint={2608.02530},
      archivePrefix={arXiv},
      primaryClass={eess.SY},
      url={https://arxiv.org/abs/2608.02530}, 
}

\end{document}